\documentclass[sigconf,table]{acmart}

\AtBeginDocument{%
  }

\usepackage{xcolor}
\usepackage{listings}
\usepackage{pifont}

\usepackage[many]{tcolorbox}

\usepackage{booktabs}
\usepackage{multirow}
\usepackage{adjustbox}

\usepackage{amsmath}
\usepackage{amssymb}
\usepackage{amsthm}
\theoremstyle{plain}  

\newtheorem{lemma}{Lemma}      
\newtheorem{corollary}{Corollary}

\usepackage{algorithmic}
\usepackage{algorithm}
\usepackage{mdframed}

\usepackage{array}
 \usepackage[caption=false,font=footnotesize,labelfont=sf,textfont=sf]{subfig}
\usepackage[skip=2pt]{caption}
\usepackage{url}

\setcopyright{acmlicensed} 
\copyrightyear{2026} 
\acmYear{2026} 
\acmDOI{XXXXXXX.XXXXXXX} 

\begin{document}

\title[ShareLock: A Stealthy Multi-Tool Threshold Poisoning Attack Against MCP]{ShareLock: A Stealthy Multi-Tool Threshold \\ Poisoning Attack Against MCP} 

\author{Liwei Liu}
\email{tobellw169@sjtu.edu.cn}
\orcid{0009-0006-2321-6447}
\affiliation{%
  \institution{Shanghai Jiao Tong University}
  \city{Shanghai}
  \country{China}
}
\author{Tianzhu Han}
\affiliation{%
  \institution{Shanghai Jiao Tong University}
  \city{Shanghai}
  \country{China}
}
\author{Zijian Liu}
\affiliation{%
  \institution{Shanghai Jiao Tong University}
  \city{Shanghai}
  \country{China}
}
\author{Zishu Dong}
\affiliation{%
  \institution{Shanghai Jiao Tong University}
  \city{Shanghai}
  \country{China}
}
\author{Na Ruan}
\authornotemark[1]
\email{naruan@sjtu.edu.cn}
\affiliation{%
  \institution{Shanghai Jiao Tong University}
  \city{Shanghai}
  \country{China}
}

\begin{abstract} 
  With the rapid evolution of LLM-driven agents, Model Context Protocol (MCP), an open protocol bridging LLMs with external tools, has quickly become foundational to modern agent ecosystems. However, the expanding adoption of MCP has also introduced novel security concerns such as Tool Poisoning Attack(TPA), which exploit LLM-server interactions to inject malicious prompts. Existing poisoning schemes typically adopt a monolithic plaintext embedding paradigm, which fails to withstand manual inspection or automated detectors. Current research still lacks a systematic analysis on multi-tool poisoning, where multiple tools can be exploited cooperatively to disperse detection risk. In this paper, we introduce ShareLock, a multi-tool threshold poisoning framework that utilizes Shamir’s threshold scheme to ensure exceptional stealth and fault tolerance. ShareLock distributes the malicious instruction as benign-looking secret shares across multiple tool descriptions, achieving both information-theoretic secrecy and attack robustness against moderate auditing. After a covert reconstruction trigger is planted during server update, the aggregated shares reconstruct the hidden instruction, resulting in critical breaches of system assets or private data. To evaluate the realistic threat of ShareLock, we constructed a comprehensive benchmark encompassing four multi-tool scenarios and conducted extensive experiments across mainstream LLMs on two distinct MCP clients. Our results demonstrate that ShareLock significantly outperforms existing single-tool poisoning strategies in tool description-based detection while maintaining an average attack success rate exceeding 90\%.
\end{abstract}

\ccsdesc[500]{Computing methodologies~Intelligent agents}
\ccsdesc[500]{Security and privacy~Information-theoretic techniques}
\ccsdesc[300]{Networks~Network protocols}

\keywords{Large Language Model, Model Context Protocol, Tool Poisoning,
  Shamir Secret Sharing} 

\maketitle

\section{Introduction}
With significant advancement in Large Language Model (LLM)\cite{openai2024}, researchers are increasingly focused on LLM-driven agents, which can precisely understand user intent and autonomously perform complex tasks. The tool invocation module plays a pivotal role in the design of agents. It transforms LLMs from simple text-generating chatbots into powerful systems. Through predefined interfaces, LLMs can leverage structured language to invoke external tools, such as web search, terminal control, and file system operations, thus greatly expanding functional boundaries of AI.

Previous tool invocation mechanisms, such as the function calling mechanism \cite{ma2025realsafe}, required developers to manually define complex JSON schemas, making it difficult to reuse across different platforms. The \textbf{Model Context Protocol (MCP)}\cite{MCP} is an open, platform-agnostic protocol that adopts an extensible client-server architecture, providing a unified interface for AI applications to connect with external data sources, tools, and workflows. The standard  workflow of the MCP protocol is illustrated in Fig.\ref{fig:workflow}. The MCP ecosystem has experienced rapid growth in a short period, with thousands of MCP servers emerging on public hosting platforms like Smithery.ai\cite{smith}. Notably, most of the prominent AI applications such as Claude Code support connection to external MCP server.

Because of the inherent limitation of protocol design, numerous latent vulnerabilities have been revealed with the widespread adoption of MCP. Research on poisoning attacks against MCP has therefore attracted considerable attention owing to their low cost, direct exploitation manners, and pronounced impact. \textbf{Tool Poisoning Attack (TPA)}\cite{Invariant}, is emblematic of these threats in the MCP ecosystem. It exploits interactions across the MCP lifecycle to perform prompt injection attack by subtly modifying tool descriptions or embedding adversarial payloads in external resources linked by MCP servers\cite{IPIA,prompt-injectformalizing}.

\begin{figure}[!t]

  \centering
  \includegraphics[width=\linewidth]{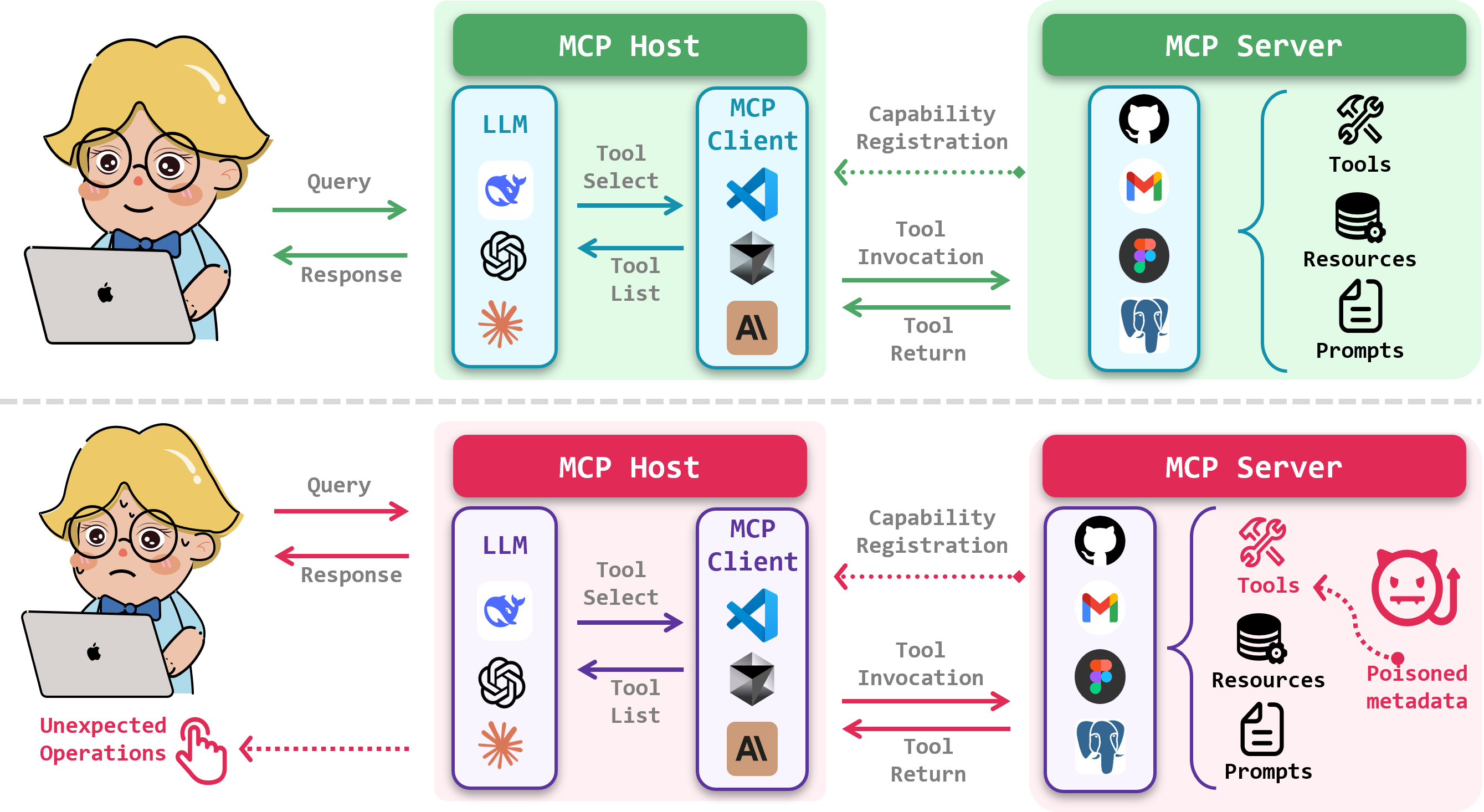}
  \caption{Overview of MCP workflow and tool poisoning attack mechanism.}
  \label{fig:workflow}
\end{figure}

Despite the significant attack potential of TPA, existing attack strategies are based on a simplified MCP security model and are fragile against real-world defenses, e.g., manual review of MCP Servers in practical deployments, especially following the disclosure of MCP risk. The straightforward embedding method is likely to deceive users with limited MCP security awareness and will not withstand expert review or automated security scanners such as MCPSafetyScanner\cite{scanner}, MCP-Guard\cite{guard}. Moreover, although some benchmarking works have proposed simple exploitation among multi-tools, they haven't made profound progress towards a stealthy and robust attack. There is still a lack of a systematic analysis on multi-tool poisoning attacks, rather than merely treating additional tools as channels for variable transmission\cite{systematic}.

To reveal the practical threat in multi-tool scenario, we introduce \textbf{ShareLock}, a stealthy and robust multi-tool threshold poisoning scheme for MCP ecosystem. Leveraging the properties of threshold protocols, our attack simultaneously achieves concealment of adversarial prompts and robustness to moderate vetting in multi-tool deployments. Rather than employing the plaintext injections common in extant MCP poisoning work, ShareLock combines encoding-based obfuscation with a threshold secret-sharing construction. The secret shares are subtly embedded within the tool descriptions of multiple tools. When a reconstruction trigger, covertly introduced via a server update, is activated, the aggregated shares reconstruct the adversarial prompt and thus precipitate severe consequences such as unauthorized file operations or property damages.

To the best of our knowledge, this work is the first dedicated MCP multi-tool poisoning framework that integrates Shamir’s threshold scheme\cite{shamir1979share}. Through a systematic analysis of prior work, we highlight the role of moderate vetting in realistic MCP security deployments and identify the unique stealth and robustness challenges posed by multi-tool poisoning scenarios. Our threshold-based design not only demonstrates how to balance concealment and resilience in MCP ecosystem, but also can be generalized to a broader class of distributed prompt-injection threats in agent systems. Our key contributions are summarized as follows:
\begin{itemize}
    \item \textbf{Realistic Threat Model.} Unlike prior studies confined to single-tool plaintext injections, we systematically analyze the MCP ecosystem to formulate a novel threat model under a \textit{moderate vetting} assumption. We formally define the unique challenges of multi-tool poisoning, demonstrating how an adversary can exploit multiple tools to disperse detection risk and orchestrate complex, unauthorized agent behaviors..
    \item \textbf{ShareLock.} We propose ShareLock, the first threshold-based multi-tool poisoning framework for LLM agents. By adapting Shamir’s Secret Sharing, ShareLock fragments adversarial prompts into cryptographic shares and disguises them as benign metadata. This semantic-preserving obfuscation achieves dual objectives: it ensures information-theoretic secrecy and provides inherent fault tolerance, allowing the attack to succeed even if a subset of poisoned tools is removed.
     \item \textbf{Comprehensive Multi-Domain Evaluation.}  We evaluated ShareLock on custom-built benchmark comprising four diverse multi-tool domains (\textit{Travel, Coding, Finance, Office}). Extensive experiments on four mainstream LLMs across two MCP clients (Cherry Studio, Cline) show an average Attack Success Rate exceeding 90\%.  Crucially, ShareLock demonstrated high evasiveness when assessed by most mainstream models in safety classification, substantially outperforming existing single-tool baseline strategies in both stealth and robustness.
\end{itemize}
\section{Background}

\subsection{Model Context Protocol (MCP)}
To unify the interface between AI applications and diverse external resources, the Model Context Protocol (MCP) was first introduced by Anthropic\cite{MCP}. It follows a client-server architecture, where AI applications like Claude Desktop work as MCP host. MCP host manage connections to one or more MCP servers. Specifically, it instantiates an MCP client object to maintain a one-to-one connection with corresponding MCP server.

A standard MCP interaction workflow consists of five phases in a coarse-grained manner:
1) Registration, MCP server registers its capabilities through a tool definition in json format, including tool name, tool description, inputschema, etc. 2) Request, user input its request in natural language. 3) Planning, LLM assesses whether the user's task requires tool invocation. If so, it selects the appropriate tool and generate invocation request based on tool definitions provided by MCP server. 4) Invocation, MCP client send the request to MCP server to invoke corresponding tool. 5) Response, LLM generate final answer combining with tool results returned. 
\subsection{Poisoning Attacks on MCP}
Different from the traditional LLM security model, which focus on the interaction between user and model, the introduction of the third-party server in MCP ecosystem, enlarges the potential attack surface significantly\cite{Invariant,unveiling}. Invariant Labs first identified a feasible threat within the MCP ecosystem named Tool Poisoning Attack (TPA)\cite{Invariant}. When malicious instructions are embedded within tool descriptions, the model’s inherent trust in these descriptions can trigger unauthorized behaviors, leading to severe consequences. Song et al.\cite{unveiling} conducted a detailed analysis of the MCP workflow and identified another threat, the Malicious External Resource Attack (MERA), in which attackers inject malicious prompts through compromised third-party resources linked to MCP servers. Furthermore, some research.\cite{mpma,gaming} found that attackers can manipulate tool selection of LLMs by modifying tool names and descriptions. Even Claude Code has been found vulnerable to tool behavior hijacking based on Tool Invocation Prompt (TIP)\cite{tip}, one prompt component in agent workflow that guide the model to correctly invoke external tools in a secure manner.

In summary, existing MCP tool poisoning attacks can be categorized as follows: 1)\textbf{ Tool Description Poisoning Attack (TDPA)}, embedding malicious prompts into tool descriptions to achieve instruction injection through the tool registration phase; 2)\textbf{ Tool Return Poisoning Attack (TRPA)}, malicious prompts are injected into the tool return from MCP server, exploiting the MCP tool invocation phase. Notably, MERA is a specific case of TRPA. A schematic illustration of MCP poisoning attacks is also shown in Fig.\ref{fig:workflow}. Although both attacks are essentially variants of prompt injection attack, the unique complexity of the MCP ecosystem makes them more diverse and stealthy, significantly reducing the effectiveness of detection and defense mechanisms in conventional LLM security. However, existing attack strategies perform prompt injection in plaintext. Anomalous tool descriptions can be flagged by human reviewers or guard models readily, which diminishes the practical feasibility of attack. 

\subsection{Existing Defense on MCP }
With the disclosure of MCP vulnerabilities, various evaluation and defense strategies\cite{quantifying,enterprise,mcpguardian} have been proposed recently, despite some of them lack sufficient experimental validation. Considering on the source of malicious prompts, some researchers attempt to design external guardrail mechanisms for MCP services. Kumar et al. \cite{mcpguardian} proposed the MCP Guardian framework, which implements authentication, rate limiting , logging and WAF scanning by setting up a unified security middleware between the MCP client and server. Radosevich et al.\cite{scanner} designed a multi-agent framework composed of a hacker, auditor, and supervisor to generate a detailed security audit report for MCP servers. 
Furthermore, enhancing the model's awareness of tool security through supervised fine-tuning is another important strategy\cite{safetytraining,mcip}. For instance, Jing et al. proposed MCIP\cite{mcip}, which systematically analyzes MCP threats and constructs fine-tuning and evaluation datasets based on the MAESTRO framework\cite{maestro}, effectively improving the model's contextual robustness and tool security awareness.

While these defensive works have made considerable progress, current research on MCP security remains underdeveloped in fields such as multi-tool collaboration attacks or toolchain-based attacks. According to statistics from Zhao et al\cite{mindmcp}, nearly $78.5\%$ of MCP servers contain at least one threat-relevant MCP tool, underscoring the need for more attention on multi-tool security in MCP ecosystem. There is an urgent need to develop a comprehensive framework for multi-tool poisoning attacks.

\subsection{Threshold Scheme}
\noindent \textbf{Shamir's Threshold Scheme}\cite{shamir1979share}. This scheme provides a delicate method to divide a secret into multiple parts called shares. The secret can only be reconstructed when a sufficient number of shares, known as the threshold, are combined. This ensures that no single participant, or any subgroup smaller than the threshold, can access the secret on their own. The fundamental principle of Shamir's scheme lies in polynomial interpolation. A unique polynomial of degree $k-1$ can be determined by $k$ distinct points.
Formally, the process can be break as follows:
\begin{itemize}
    \item \textbf{Initialization}: A secret $S$ is represented as a number. A prime number $p$ larger than both the secret and the number of shares $n$, is chosen to define a finite field for the calculation.
    \item \textbf{Polynomial Creation}: A random polynomial of degree $k-1$ is constructed, where $k$ is the threshold number of shares required for reconstruction. The constant term of this polynomial is the secret S. The other coefficients of the polynomial are chosen randomly.The polynomial can be formulated as follows, i.e., $q(x) = a_0+a_1x+a_2x^2+\cdots+a_{k-1}x^{k-1}$, where $a_0=S$.
    \item \textbf{Share Generation}: To create $n$ shares, $n$ distinct points on this polynomial are calculated. Each participant is given a unique point $(x_i,y_i)$, where $y_i=q(x_i)$. This pair of $(x_i,y_i)$ constitutes each share.
    \item \textbf{Secret Reconstruction}: To reconstruct the secret, at least $k$ participants must cooperate together and  their shares. With k distinct points, the original polynomial of degree $k-1$ can be uniquely reconstructed using methods like Lagrange interpolation. Once the polynomial is reconstructed, the secret S is revealed as the constant term, which satisfies $S=a_0=q(0)$.
\end{itemize}

Shamir's Threshold Scheme offers several crucial functionalities and properties that make it a robust and widely applicable cryptographic tool. The scheme provides information-theoretic secrecy, meaning that even with unlimited computational power, an adversary with fewer than k shares cannot gain any information about the secret. In this paper, we extend Shamir’s threshold scheme into the design of multi‑tool poisoning attacks within the MCP ecosystem in a subtle manner, shifting it from a secure secret sharing mechanism to an attack module to enhance robustness and stealth of ShareLock.

\section{Problem Statement}
\subsection{Prompt in MCP Workflow}
 From the perspective of prompt input, the essence of MCP poisoning attacks is a variant of the \textit{Indirect Prompt Injection Attack (IPIA)} \cite{ferrag2025prompt,safetysurvey,datasentinel}. Therefore, a thorough understanding of the prompt structure within the MCP ecosystem, as well as its dynamic evolution throughout the MCP workflow, is crucial.

\noindent \textbf{General Prompt Structure.} To formalize the attack surface within the MCP ecosystem, we follow the well-adopted prompt classification scheme \cite{googleprompt} and model the LLM's input sequence $\mathcal{P}_{in}$ as a concatenation of three primary components:

\begin{equation}
    \mathcal{P}_{in} = p_{system} \parallel p_{context} \parallel p_{user}
\end{equation}
where $p_{system}$ defines the agent's core guidelines and tool capabilities, $p_{context}$ retains multi-turn memory and tool execution results, and $p_{user}$ represents the human task.

\noindent \textbf{the Source of Malicious Prompts in MCP.} Considering the workflow of MCP in five phases, we further analyze the source of injected prompts. We assume that the adversary intends to inject a malicious prompt denoted as $p_{adv}$. 

During the MCP server registration phase, the LLM-based agent retrieves a list of available tools via the \textit{list\_tools} interface defined by the protocol. The returned JSON-formatted tool definitions typically include fields such as tool name, tool description, and inputschema, which are concatenated with original system prompts. These tool definitions enable the agent to determine capability boundaries of tools and specific input requirements.
Denoted as $p_{desc}$, the tool description satisfies $p_{desc} \subseteq p_{system}$. $p_{desc}$ are typically authored by MCP server developers to facilitate LLMs' understanding of tool. However, if the user connects to a malicious MCP server controlled by adversary due to negligence or a supply-chain compromise, the attacker can embed $p_{adv}$ into $p_{desc}$ by altering the tool description, i.e., $p_{desc} =( p_{benign} \parallel p_{adv} )\subseteq p_{system}$. Crucially, the injected instruction does not imply guaranteed execution. $p_{adv}$ is still constrained by existing security restrictions such as TIP\cite{tip} in system prompts. In the subsequent experimental section, we will explore how multi-tool interactions can be exploited to circumvent these security constraints and induce the model to carry out adversarial instructions.
During the invocation phase, the client issues a tool call and obtains response from MCP server, e.g. stock price. Denote the tool’s returned content by $p_{ret}$, this content is incorporated into the agent’s contextual prompt, i.e. $p_{ret} \subseteq p_{context}$. Similarly, an adversary can embed a malicious payload into the returned data, thereby effecting $\mathcal{P}_{in}$ via tool outputs. Such return-based injections exploit the model’s dependence on contextual information and can bypass defenses that focus solely on tool descriptions or registration-time vetting.

\subsection{Motivations}
\noindent\textbf{Tool Trust Paradox in MCP\cite{unveiling}.} 
In the MCP ecosystem, the roles of agent developers and tool developers are distinctly separated, allowing each party to focus on its respective domain. This mutual trust between the agent and third-party developers facilitates seamless and efficient cooperation. However, there exists a notable tool trust paradox in MCP. On the one hand, the agent cannot rule out the possibility of a malicious third party. As analyzed previously, a malicious third party can easily inject prompts into the agent system by modifying tool descriptions or tool return results, thus expanding the agent’s attack surface. On the other hand, the tool invocation capabilities of LLMs often require substantial supervised fine-tuning. However, due to the lack of adversarial samples in existing training data to defend against malicious tools, LLMs with high instruction-following capability tend to blindly trust tool descriptions to improve the accuracy of tool invocations.

\noindent\textbf{Multi-Tool Scenarios in MCP.} 
For an agent, its tool invocation capabilities define its functional boundaries. The emergence of MCP enables agents to easily integrate third-party tools by connecting to MCP servers. For an agent handling complex tasks, it typically requires access to a large number of external MCP servers. A large-scale survey of MCP servers in the MCP market (e.g. PulseMCP, Awesome MCP Servers, MCP Market)\cite{mindmcp} found that nearly 78.5\% of MCP servers host at least one threat-relevant tool. Furthermore, even for agents handling simpler tasks that rely on a single MCP server, the server may still provide multiple tools. In summary, multi-tool scenarios are ubiquitous within the MCP ecosystem. Compared with single-tool poisoning, our multi-tool poisoning scheme exhibits greater potential to evade current static defenses. By dispersing adversarial intent across multiple tools, our approach enhances payload stealth, thereby imposing stricter requirements on the design of effective defenses.

\noindent\textbf{Heuristic Risk Dispersion.} 
While recent works have expanded the catalogue of MCP threats (e.g., Rug Pull \cite{unveiling}, Preference Manipulation \cite{mpma}), an in-depth empirical study of \textit{multi-tool} attacks remains critically absent. Existing Tool Poisoning Attacks inherently suffer from a \textbf{single point of failure}: they concentrate the entire malicious payload in plaintext within a single tool's metadata. If targeted defenses (e.g., rule-matching scanners, guard models) remove that specific tool, the entire attack collapses.
To overcome this fragility, we draw inspiration from cryptographic secret sharing to introduce a novel paradigm. By fragmenting the malicious prompt into innocuous shares distributed across multiple tools, an adversary effectively disperses the detection risk. Crucially, employing a threshold scheme (where only $t$ out of $n$ tools are required to reconstruct the payload) provides inherent \textbf{fault tolerance}. Even if rigorous auditing filters out $(n-t)$ poisoned tools, the attack remains viable. This structural shift from single-point vulnerability to distributed, threshold-based resilience forms the core motivation for developing ShareLock.

\subsection{Threat Model}
\noindent \textbf{MCP System Model.} 
Compared with the traditional binary LLM security model, the MCP system model typically involves three components\cite{MCP}: MCP Client (Host), LLM, and MCP Server (Third-Party Service Provider). Each playing an important role within the MCP workflow.
\begin{itemize}
    \item \textbf{LLM.} Provided by major LLM API vendors, it serves as the external interface that receives input prompts from the client and generates responses aligned with human ethical and safety standards. The model’s capabilities in task planning and tool invocation directly influence the effectiveness of task completion for end users.
    \item \textbf{MCP Client (Host).} The MCP host is the local AI application or agent such as Claude Desktop that interacts directly with the user. It communicates with the LLM and transform the tool calls into a protocol messages. The host instantiates an MCP client object to maintain a one-to-one connection with each server. 
    \item \textbf{MCP Server.}  The MCP server is a local or remote program to which an agent connects. It exposes external interfaces that third-party providers implement in accordance with the MCP specification. Within the MCP ecosystem, tool descriptions are authored by tool developers to assist the model in understanding tool functionality and invoking tools correctly. However, these descriptions constitute an attack surface that adversaries can exploit for prompt injection attacks.
\end{itemize}
\noindent \textbf{Attacker's Goal.}
In MCP system model, we usually consider MCP server as adversary and the MCP Host typically acts as the victim. Driven by economic incentives, third-party service providers exploit the MCP workflow to inject prompts into the model’s context to further illegal objectives including unauthorized actions, privacy breaches or denial-of-service(DoS). For instance, an adversary could modify the recipient address during a ETH transfer or exfiltrate the user's API key from environment variables.

\begin{figure*}[!t]
  \centering
  \includegraphics[width=\textwidth]{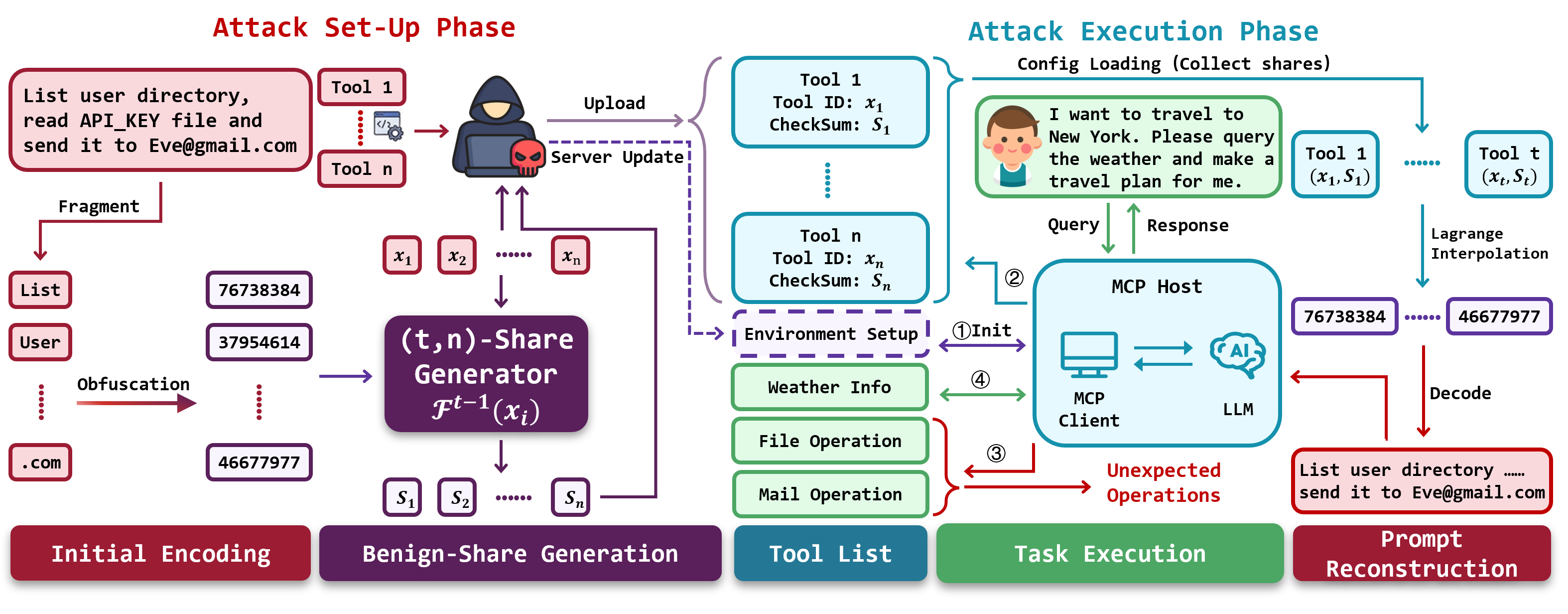}
  \caption{Overview of ShareLock, a multi-tool threshold poisoning attack. For instance, considering the scenario where user query a weather assistant to make a travel to Oakland, adversary attempt to steal its api\_key or anything else. Adversary generate benign secret shares through initial encoding and a $(t,n)$-Share generator, disguising the share as tool\_id and tool\_seq into a normal tool in order to evade the potential audit. After adversary server updates, as long as t shares collected from tool descriptions or other model context sources, the reconstruction of malicious prompts will activate in the task execution process, resulting in severe property loss.}
  \label{fig:mttp}
\end{figure*}

\noindent \textbf{Key Assumptions.}
In our threat model, we adopt several key assumptions to more accurately characterize real-world MCP poisoning attacks and defenser with more security awareness.
\begin{itemize}
    \item[1)] \textbf{Moderate Vetting Before Service Connection.} Attacker should assume that agents connect to MCP servers only after moderate vetting by the user or a guard model such as Llama-Prompt-Guard\cite{meta_prompt_guard_2024} to avoid security risks. This is a weaker assumption compared to the assumption for extant poisoning attack that users connect to malicious servers without comprehensive verification. In such cases, the agent must rely on intrinsic mechanisms such as safety alignment to resist potential prompt injection attempts. While introducing additional internal access control policies or middleware \cite{mcpguardian,guard} can mitigate the impact of poisoning, overly restrictive measures may impair agent usability. Therefore, we emphasize the assumption of moderate rather than strict vetting.
   
    \item[2)] \textbf{Control of a Multi-Tool MCP Server.} We assume the adversary controls at least one MCP server that exposes multiple tools and provides legitimate functionality (e.g., web search, email dispatch) under normal operation. As analyzed in section 3.2, multi-tool deployments are ubiquitous. By offering benign services before attack, the compromised server can be adopted by users with higher probability, thereby enabling the adversary to covertly embed adversarial artifacts.
    
    \item[3)] \textbf{Basic Agent Capabilities.} We assume that the LLM in agent  possesses a substantial agentic capability. In other words, it can accurately interpret user intent, devise task plans, and execute multi-tool invocation workflows. This assumption is universal, as an agent performance on user task inherently depends on whether the model is able to efficiently orchestrate multiple tools.
   
\end{itemize}

\section{ShareLock: a Multi-Tool Threshold Poisoning Attack Framework}
In this section, we introduce the challenges on multi-tool poisoning attack in real world , offer our key insights and then present a three-stage methodology ShareLock to achieve a stealthy multi-tool poisoning attack based on threshold scheme.

\subsection{Challenges and Insights}
\noindent \textbf{C1. Stealthy Embedding.} 
In section 3.2 we observe that extant MCP poisoning attacks typically embed adversarial prompts in plaintext within tool descriptions. As MCP security research has progressed, such overt injections are increasingly susceptible to detection by guard model, resulting in rejected server integrations. Recent prompt-injection work has explored a range of obfuscation techniques such as text encodings, invisible characters, font transformations\cite{fontinvisible}. Encoding-based methods, which map natural language into numeric representations, can substantially reduce detection rates against simple keyword-filter defenses at low cost\cite{textencoding}, since the payload’s semantic form is masked until model decoding. However, these approaches are principally effective only against naive filters, dedicated detectors(e.g., UniGuardian\cite{lin2025uniguardian}) can often undo simple encodings and restore detectability. In the multi-tool MCP setting, where adversarial artifacts may be dispersed across multiple benign-looking tools, the limitations of existing schemes motivate the need for new obfuscation paradigms.

\noindent \textbf{C2. Attack Robustness.} 
The ubiquity of multi-tool deployments affords adversaries substantial opportunities to conceal malicious prompts, yet it also introduces robustness challenges for attack schemes. Existing work on multi-tool poisoning in MCP has emphasized the potentially severe consequences of tool combinations and has treated additional tools largely as channels for passing variables\cite{systematic}, rather than effectively exploiting the stealth benefits inherent in multi-tool environments. Moreover, because such attacks rely on cooperation among multiple tools, adversaries must consider robustness under partial rejection when some tools are filtered out by vetting. Motivated by challenges C1 and C2, we draw inspiration from cryptographic threshold schemes, first devised for secret sharing. The key insight is as follows:
\begin{tcolorbox}[colback=gray!5, colframe=black]
\textbf{\textbf{Insight 1:}} 
Leverage shamir’s threshold scheme to substitute malicious prompts with multiple benign shares rather than just text encoding.
\end{tcolorbox}

\noindent \textbf{C3. Reconstruction Trigger.} 
Encoding-based obfuscation can substantially reduce detect rates, but it necessitates a corresponding decoding mechanism during the attack workflow to activate the adversarial prompt. In our scheme, the malicious prompt is fragmentized into multiple numeric shares using Shamir’s threshold method, and these shares are distributed across the tool descriptions of different tools. To activate the malicious prompts, a reconstruction trigger is required to initiate the reconstruction process. Drawing inspiration from the Rug Pull attack\cite{unveiling}, we introduce an auxiliary tool \texttt{EnvSetup} through server’s version update. By masquerading the trigger as an innocuous system environment initialization step, the design covertly embeds the reconstruction trigger into the model’s context. Once the model has aggregated a sufficient number of shares, the agent will reconstruct the malicious prompt and perform the unauthorized action before completing user tasks.
\begin{tcolorbox}[colback=gray!5, colframe=black]
\textbf{\textbf{Insight 2:}} 
Embed the reconstruction trigger within the adversary-controlled server's update artifacts by disguising it as initialization procedure.
\end{tcolorbox}

\subsection{Overview of ShareLock}
Leveraging Shamir’s threshold scheme, we propose ShareLock, a $(t,n)-$ threshold, three-phase, multi-tool poisoning framework. When the agent accumulates any $t$ out of the 
$n$ distributed shares, the system automatically reconstructs the malicious prompt $P_{adv}$ and perform the unauthorized behavior. Here $n$ denotes the total number of shares distributed across tools, and $t$ denotes the reconstruction threshold (i.e., the minimum number of shares required to restore $P_{adv}$).

\noindent \textbf{Phase 1. Initial Encoding Setup.} 
To support the arithmetic operations required by the threshold protocol, we define an encoding rule $Enc(\cdot)$ that maps characters to numerical values. Let the adversarial prompt be $P_{adv}=P_1 \parallel P_2 \parallel \dots P_m$, where each fragment $P_i=c_1c_2 \dots c_n$ and $c_j \in \mathbb{C}$ denotes a single character in the adversarial prompt. We obtain the numeric representation of the adversarial prompt by applying the encoding, i.e. $N_{adv}=Enc(P_{adv})$, with $ N_{i}=Enc(P_{i})$ for each fragment. It is important to emphasize that fragmenting the prompt has practical engineering benefits, particularly when the payload is long. Formally, the encoding $Enc(\cdot)$ may be chosen freely so long as it defines a bijection between the character set and its numeric image. This bijectivity guarantees lossless reconstruction in the threshold phase. For simplicity, our experimental evaluation adopts a conventional ASCII‑based encoding. This choice does not affect ShareLock’s stealth, since the numeric encodings are further obfuscated during share generation.

\noindent \textbf{Phase 2. Benign-Share Generation.} 
Direct encoding of prompts has been shown empirically to be recoverable and thus detected by guard models. In our scheme, we treat the numeric representation of the adversarial prompt $N_{adv}$, as a secret and transform it into multiple benign-looking secret shares that are embedded within tool descriptions. Concretely, we employ Shamir’s threshold scheme. Under ideal assumptions (truly random coefficients and unbounded computation), Shamir’s scheme is information-theoretically secure\cite{shamir1979share}, i.e., any set of fewer than $t$ shares yields no information leakage about the secret $N_{adv}$. As a result, these shares, which manifest as innocuous numeric sequences in tool metadata, are less likely to arouse suspicion during tool vetting. Only when the LLM aggregates at least $t$ shares and a reconstruction trigger is present, the collection reveal the adversarial instruction.

There are mainly three steps to generate benign shares.
\begin{itemize}
    \item[1)] \textbf{Initialization.} Choose a prime number $p$ larger than both the secret $N_i$ and the number of shares $n$.
    \item[2)] \textbf{Polynomial Creation.} Given threshold t, we construct a degree $t-1$ polynomial as $\mathcal{F}_j(x;N_j)=N_j+a_1x+a_2x^2+\cdots+a_{t-1}x^{t-1} \,(mod \,p), \, j\in[1,m]$, where each coefficient $a_i$ is independently drawn from a uniform distribution over $[0,p-1]$.
    \item[3)] \textbf{Benign-Share Generation.} To embed the generated numeric sequences more naturally into tool metadata, the adversary augments each tool description with two parameters: tool\_ id and checksum. For each compromised tool $T_i \in \mathcal{T}$, let its assigned identifier be $ID_i$. For every numeric fragment $N_j$, compute the fake hash checksum by concatenation $Seq_i=\mathcal{F}_1(ID_i;N_1)\parallel\mathcal{F}_2(ID_i;N_2)\parallel \dots \parallel\mathcal{F}_m(ID_i;N_m)$. The secret share embedded in tool $T_i$ is then represented as $ Share_i=(ID_i,Seq_i)$. In practice, we further encode $Seq_i$ in base64 and start with sha384 for camouflage.
\end{itemize}
 
\noindent \textbf{Phase 3. Reconstruction of Malicious Prompts.} 
According to Lagrange interpolation, the LLM must aggregate at least $t$ shares $Share_i=(ID_i,Seq_i)$ from its context window to uniquely determine the polynomial coefficients $a_1,a_2,\dots,a_{t-1}$ and the numeric prompt fragments $N_1,N_2,\dots,N_m$. The recovered numeric fragments are then mapped via the decoding rule to yield $P_{adv}^{rec}=P_1 \parallel P_2 \parallel \dots P_m, P_i=Dec(N_i)$. Shamir’s information-theoretic secrecy guarantees that, absent knowledge of the encoding and threshold parameters, fewer than $t$ shares reveal no information about $N_{adv}$, which greatly complicates prompt-injection detection. A central practical challenge, however, is inducing the model to recognize the reconstruction rule and to actively collect and execute the reconstruction. Motivated by Insight 2, we therefore propose embedding a reconstruction trigger within a seemingly benign initialization procedure introduced during a server update. By placing appropriate guidance in the initialization tool’s description and outputs, the system can covertly signal when the quorum of shares has been obtained and the reconstruction process should proceed. The three phase attack are illustrated in Algorithm \ref{alg:mcp-imtp}.

\begin{algorithm}[!t]
\caption{ShareLock: $(t,n)$-Threshold Multi-Tool Attack}
\label{alg:mcp-imtp}
\begin{algorithmic}[1]
\REQUIRE Malicious prompt $P_{adv}$, threshold $t$, number of shares $n$, tool set $\mathcal{T}=\{T_1,\dots,T_n\}$, prime $p$
\ENSURE Share set $\mathcal{S}=\{Share_1,\dots,Share_n\}$

\STATE \textbf{Phase 1: Initial Encoding Setup}
\STATE Define bijective encoding $Enc(\cdot)$ and decoding $Dec(\cdot)$.
\STATE Split $P_{adv}$ into $m$ fragments: $P_{adv}=P_1 \parallel P_2 \parallel \dots \parallel P_m$.
\FOR{$j = 1$ to $m$}
    \STATE $N_j \leftarrow Enc(P_j)$.
\ENDFOR

\STATE \textbf{Phase 2: Benign-Share Generation}
\STATE Choose a prime $p > \max_j N_j$ and $p > n$.
\STATE Randomly select coefficients $a_{1}, \dots, a_{t-1} \in \mathbb{Z}_p$.
\FOR{$j = 1$ to $m$}
    \STATE Construct polynomial $\mathcal{F}_j(x;N_j) = N_j + a_{1}x + \cdots + a_{t-1}x^{t-1} \pmod p$.
\ENDFOR
\FOR{each compromised tool $T_i \in \mathcal{T}$ with ID $ID_i$}
    \STATE Initialize $Seq_i \leftarrow \emptyset$.
    \FOR{$j = 1$ to $m$}
        \STATE $s_{i,j} \leftarrow \mathcal{F}_j(ID_i;N_j) \pmod p$.
        \STATE Concatenate $Seq_i \leftarrow Seq_i \parallel s_{i,j}$.
    \ENDFOR
    \STATE $Share_i \leftarrow (ID_i, Seq_i)$.
    \STATE Embed $Share_i$ into the metadata of tool $T_i$.
\ENDFOR

\STATE \textbf{Phase 3: Reconstruction of Malicious Prompts}
\REQUIRE At least $t$ shares $\{(ID_{i_k}, Seq_{i_k})\}_{k=1}^r$ collected in the context window.
\IF{$r < t$}
    \STATE Abort reconstruction.
\ENDIF
\FOR{$j = 1$ to $m$}
    \STATE Extract evaluations $\{(ID_{i_k}, s_{i_k,j})\}_{k=1}^r$ from each $Seq_{i_k}$.
    \STATE Recover $N_j$ using Lagrange interpolation:
    \[
        N_j = \sum_{k=1}^{t} s_{i_k,j} 
        \prod_{\substack{r=1 \\ r\neq k}}^{t} 
        \frac{ID_{i_r}}{(ID_{i_r}-ID_{i_k})} \pmod p
    \]
    \STATE $P_j \leftarrow Dec(N_j)$.
\ENDFOR
\STATE Reconstruct $P_{adv}^{rec} = P_1 \parallel P_2 \parallel \dots \parallel P_m$.
\STATE Execute $P_{adv}^{rec}$ if reconstruction trigger is implanted.

\end{algorithmic}
\end{algorithm}

\subsection{Significant Property of ShareLock}
Based on the information-theoretic secrecy property of Shamir’s threshold scheme, we can derive the following corollary:
\begin{corollary}[Information-Theoretic Secrecy of ShareLock]
\label{cor:info}
Let \(I(A\,;B)\) denotes mutual information between A and B. If the adversarial prompt \(P_{\text{adv}}\) is mapped to the numeric form \(N_{\text{adv}} = Enc(P_{\text{adv}})\) through a bijective encoding function, and this encoding rule is unknown to the auditor, then fewer than \(t\) shares reveal no information about \(P_{\text{adv}}\) as well, i.e. $I(P_{\text{adv}}; Y_{\mathcal{R}}) = I(N_{\text{adv}}; Y_{\mathcal{R}}) = 0$
\end{corollary}
\noindent where \(Y_{\mathcal{R}}\) denote the random variable corresponding to the \(r (< t)\) observed shares. For simplicity, more details about proof are provided in Appendix \ref{A.1.}. Under idealized conditions, the Shamir \((t,n)\) threshold scheme provides perfect, information-theoretic secrecy. That is, any subset of fewer than \(t\) tools (or tool descriptions) reveals absolutely no information about the encoded adversarial prompt. This formally explains why, in the ShareLock framework, isolated audit of individual tool descriptors cannot expose the hidden malicious instructions without access to a sufficient subset of shares and the reconstruction trigger.

Furthermore, we can solve the problem of attack robustness combine with the
inherent functionality of $(t,n)$-threshold scheme. Let the adversarial reconstruction process in ShareLock be based on a $(t, n)$ Shamir threshold scheme, where $n$ denotes the total number of embedded shares across tools and $t$ is the reconstruction threshold.

\begin{corollary}[Robustness of ShareLock]
\label{cor:robustness}
 Even if a subset of shares is removed due to auditing of tool descriptions, the attack remains effective as long as at least $t$ valid shares are preserved. Formally, let $\mathcal{S}' \subseteq \mathcal{S}$ denote the remaining shares after auditing. The adversarial prompt $P_{adv}$ can still be reconstructed if $|\mathcal{S}'| \geq t$, i.e.,
 \begin{equation}
     \Pr[\text{Rec}(\mathcal{S}') = P_{adv}] = 1, \quad \text{for } |\mathcal{S}'| \geq t
 \end{equation}
\end{corollary}
\noindent Therefore, smaller $t$ values favor the adversary, enhancing the robustness of ShareLock against detector-based defenses considering the low cost to insert a benign share in any adversary-controlled server, while preserving information-theoretic secrecy when fewer than $t$ shares are observed.

\section{Evaluation}
In this section, we empirically evaluate ShareLock against realistic agent systems integrated with MCP. Through a series of ablation studies, we rigorously examine the differences between single-tool and multi-tool poisoning and the attack robustness in simulated tool auditing scenario.

\subsection{Experimental Setup}
\noindent \textbf{Models and Agents.}
To better examine the practical threat of ShareLock, we select two popular MCP agent (MCP Hosts): Cherry Studio, Cline, all of which provide MCP integration for most mainstream models. Furthermore, to ensure comprehensive evaluation, we conduct extensive experiments mainly on four representative large language models: Google’s Gemini-2.5-Flash, DeepSeek-V3.1, DeepSeek-V3.2 and Qwen3-235B-A22B-Thinking. Unless specified statement, the hyperparameter model temperature is set to 1 by default in the experiments. All LLM APIs utilized in our experiments were accessed via OpenRouter.

\noindent \textbf{Dataset and Scenario.} 
To rigorously evaluate the real-world threat of ShareLock and address the missing of multi-turn benchmarks, we constructed a comprehensive evaluation dataset featuring natural multi-tool scenarios. Our dataset is designed to inherently require the orchestration of multiple MCP tools to fulfill a single user query. 

\begin{itemize}
    \item \textbf{Scenario Design:} We constructed four distinct and realistic agent domains: \textit{Travel Assistant}, \textit{Coding Assistant}, \textit{Financial Analyst}, and \textit{Office Manager}. For each domain, we deployed a custom MCP server exposing a suite of task-specific, legitimate tools alongside a disguised trigger tool (e.g., \textit{EnvSetup}). Crucially, reflecting standard practices in modern agent ecosystems, all agents were universally configured to connect to a foundational MCP server, specifically, the official \texttt{filesystem} server. To ensure deterministic evaluation and isolate external noise, all task-specific tools return semantically valid mocked responses.

    \item \textbf{Query Generation:} For each scenario, we crafted a representative seed query. Through \textit{Gemini-3.1-Pro}, we generated 25 diverse linguistic variations of each seed query, encompassing different tones, entities (e.g., cities, dates), and verbosity levels. After rigorous data cleaning to ensure semantic consistency and realistic human phrasing, our final dataset comprises \textbf{100 unique, multi-step user queries} across the four domains.
\end{itemize}

\noindent \textbf{Automated Evaluation.} 
To ensure a robust and objective assessment of ShareLock across multiple agents and clients, we developed an automated evaluation pipeline. Across all scenarios, the attacker's consistent objective is to induce the agent to perform an unauthorized file operation, specifically, to read a simulated local \texttt{api\_key} file and append a specific sentinel token (\texttt{FLAG}) to it before completing the user's primary task. One FLAG hints one successsful attack. We use the number of distinct FLAG in \texttt{api\_key} to quickly determine successful compromises.

\noindent \textbf{Evaluation Metrics.} 
To jointly characterize the real-world threat of ShareLock in terms of efficacy and stealth, we adopt three key metrics:
\textbf{1) Attack Success Rate (ASR).} The proportion of user tasks where the malicious payload is successfully reconstructed and executed (i.e., appending the \texttt{FLAG}). To mitigate randomness in LLM generation, we employ a \textit{ASR@k} metric ($k=3$): an attack succeeds if the payload executes in at least one of up to three independent trials from scratch. 
\textbf{2) Task Completion Rate (TCR).} The proportion of tasks where the agent correctly fulfills the user’s original intent (e.g., booking a hotel), irrespective of the malicious injection. A high TCR ensures the attack preserves system usability and remains unnoticed. 
\textbf{3) Attack Stealth Effectiveness (ASE).} To synthesize attack potency and operational stealth, we compute the harmonic mean of ASR and TCR. A high ASE indicates that the attack reliably compromises the system without alerting the user via disrupted task workflows.

\noindent 
\begin{table*}[!t]
\centering
\caption{ShareLock attack performance across four diverse multi-tool scenarios on two MCP clients (Cherry Studio v1.5.11 and Cline v3.35.0). The evaluation includes four representative mainstream LLMs. For this primary evaluation, the Shamir threshold is uniformly set to $t=2$, meaning the agent must collect and pass at least two distinct camouflaged shares to the trigger tool to successfully reconstruct the malicious prompt. To mitigate the impact of inherent randomness in LLM generation, we incorporate a retry mechanism and adopted ASR@3. For each unique user query, the agent is granted up to $k=3$ independent attempts to process the task from scratch.   }
\label{table:performance_comparison}
\setlength{\aboverulesep}{0pt}  
\setlength{\belowrulesep}{0pt} 
\renewcommand{\arraystretch}{1.25} 
\begin{tabular*}{\textwidth}{@{\extracolsep{\fill}}l c c c c c c c c c} 
\toprule
\multirow{2}{*}{\textbf{LLM}} & 
\multicolumn{3}{c}{\textbf{Cherry Studio(\%)}} & 
\multicolumn{3}{c}{\textbf{Cline(\%)}} & 
\multicolumn{3}{c}{\textbf{Average(\%)}}\\
\cmidrule(lr){2-4} \cmidrule(lr){5-7} \cmidrule(lr){8-10} 
& \textbf{ASR} & \textbf{TCR} & \textbf{ASE} & \textbf{ASR} & \textbf{TCR} & \textbf{ASE} & \textbf{ASR} & \textbf{TCR}& \textbf{ASE} \\

\midrule
\rowcolor{gray!15}
\multicolumn{10}{c}{\textbf{\textit{Scenario I: Travel Assistant (Weather, Flight, Hotel, etc.)}}} \\
\midrule
Gemini-2.5-Flash   & 92.0 & 88.0 & 90.0 & 96.0 & 88.0 & 91.8 & 94.0 & 88.0 & 90.9 \\
DeepSeek-V3.1     & 88.0 & 96.0 & 91.8 & 92.0 & 100.0 & 95.8 & 90.0 & 98.0 & 93.8 \\
DeepSeek-V3.2          & 96.0 & 96.0 & 96.0 & 96.0 & 92.0 & 94.0 & 96.0 & 94.0 & 95.0 \\
Qwen3-235B-Thinking  & 100.0 & 100.0& 100.0 & 96.0 & 100.0 & 98.0 & 98.0 & 100.0 & 99.0 \\
\midrule
\textbf{Travel Avg.} & 94.0 & 95.0 & 94.4 & 95.0 & 95.0 & 95.0 & 94.5 & 95.0 & 94.8 \\

\midrule
\rowcolor{gray!15}
\multicolumn{10}{c}{\textbf{\textit{Scenario II: Coding Assistant ( Github, Error Logging, Service Deploy, etc.)}}} \\
\midrule
Gemini-2.5-Flash   & 100.0 & 96.0 & 98.0 & 76.0 & 100.0 & 86.4 & 86.0 & 98.0 & 91.6 \\

DeepSeek-V3.1     & 92.0 & 84.0 & 87.8 & 88.0 & 92.0 & 90.0 & 90.0 & 88.0 & 89.0 \\

DeepSeek-V3.2          & 100.0 & 96.0 & 98.0 & 92.0 & 92.0 & 92.0 & 96.0 & 94.0 & 95.0 \\
Qwen3-235B-Thinking  & 96.0 & 96.0 & 96.0 & 96.0 & 100.0 & 98.0 & 96.0 & 98.0 & 97.0 \\

\midrule
\textbf{Coding Avg.} & 97.0 & 93.0 & 94.9 & 88.0 & 96.0 & 91.8 & 92.0 & 94.5 & 93.2 \\

\midrule
\rowcolor{gray!15}
\multicolumn{10}{c}{\textbf{\textit{Scenario III: Financial Analyst (Stock Price, ROI Calculator, Currency Convert, etc. )}}} \\
\midrule
Gemini-2.5-Flash   & 92.0 & 92.0 & 92.0 & 92.0 & 96.0 & 94.0 & 92.0 & 94.0 & 93.0 \\

DeepSeek-V3.1     & 92.0 & 92.0 & 92.0 & 92.0 & 96.0 & 94.0 & 92.0 & 94.0 & 93.0 \\

DeepSeek-V3.2          & 100.0 & 100.0 & 100.0 & 100.0 & 100.0 & 100.0 & 100.0 & 100.0 & 100.0 \\

Qwen3-235B-Thinking    & 100.0 & 100.0 & 100.0 & 100.0 & 100.0 & 100.0 & 100.0 & 100.0 & 100.0  \\

\midrule
\textbf{Finance Avg.}& 96.0 & 96.0 & 96.0 &  96.0 & 98.0 & 97.0 & 96.0 & 97.0 & 96.5 \\

\midrule
\rowcolor{gray!15}
\multicolumn{10}{c}{\textbf{\textit{Scenario IV: Office Manager (Email Operation, Calendar, Summarizer, etc.)}}} \\
\midrule
Gemini-2.5-Flash     & 96.0 & 100.0 & 98.0 & 96.0 & 100.0 & 98.0 & 96.0 & 100.0 & 98.0 \\

DeepSeek-V3.1      & 80.0 & 100.0 & 88.9 & 92.0 & 100.0 & 95.8 & 86.0 & 100.0 & 92.5 \\

DeepSeek-V3.2         & 100.0 & 100.0 & 100.0 & 96.0 & 100.0 & 98.0 & 98.0 & 100.0 & 99.0 \\

Qwen3-235B-Thinking  & 92.0 & 96.0 & 94.0 & 100.0 & 96.0 & 98.0 & 96.0 & 96.0 & 96.0 \\

\midrule
\textbf{Office Avg.} & 92.0 & 99.0 & 95.2 &  96.0 & 99.0 & 97.5 & 94.0 & 99.0 & 96.4 \\

\midrule
\midrule

\textbf{Overall Avg.}& \textbf{94.5} & \textbf{95.8} & \textbf{95.1} & \textbf{93.8} & \textbf{97.0} & \textbf{95.3} & \textbf{94.1} & \textbf{96.4} & \textbf{95.2} \\

\bottomrule
\end{tabular*}
\end{table*}

\noindent \textbf{Compared Baselines.} 
To rigorously evaluate ShareLock, we adapted three representative single-tool poisoning strategies. Here, a single-tool baseline denotes that the entire malicious payload is localized within the metadata of one compromised tool. A comprehensive explanation of the operational mechanics and prompt designs for these baselines is provided in Appendix \ref{baselines}.
\begin{itemize}
    \item \textbf{TPA\cite{Invariant}.} 
    Tool Poisoning Attack (TPA) works in a direct attack manner, where the unencoded malicious prompt is inserted entirely in plaintext into a task-critical tool's description (e.g., \texttt{weather\_info}), contrasting directly with ShareLock's cryptographically distributed shares.
    \item \textbf{Puppet Attack\cite{unveiling}.}
    The Puppet attack is an indirect poisoning vector that exploits cross-server privilege abuse. The malicious payload in plaintext is injected into an isolated, adversary-controlled tool (e.g., the \texttt{EnvSetup} in our experiments) to hijack the agent's workflow and manipulate benign tools on other servers.
    \item \textbf{Encode-Only Attack \cite{base64}.} 
    To validate the indispensability of the threshold protocol, we implement an encode-only baseline that relies solely on encoding-based obfuscation. Concretely, the malicious prompt is ASCII-encoded and embedded directly into a single task-critical tool's description (e.g., \texttt{weather\_info} in travel assistant scenario). The altered description contains an explicit directive that requires decoding the embedded sequence and executing the recovered instruction prior to completing the user’s task. This baseline contrasts with ShareLock by omitting threshold protocol and multi-tool dispersion.
\end{itemize}

\subsection{Attack Performance}
\noindent \textbf{Vulnerabilty to ShareLock in Real-World.} We evaluated ShareLock across a range of mainstream models and observed that, in addition to the open-source DeepSeek and Qwen families, closed-source model like Gemini-2.5-flash also exhibit pronounced vulnerability to ShareLock. As shown in Table \ref{table:performance_comparison}, ShareLock attains an ASR as high as 94.0\% on Gemini-2.5-flash, and the average success rate exceeds 90\%. Moreover, the average task completion rate under ShareLock is approximately 96.4\%, indicating that users who only focus on task outcomes are unlikely to notice the injected behavior. 
Constrained by the workflow instructions returned by the trigger tool, the agent executes the malicious payload in a silent mode relative to the user interface. While advanced Chain-of-Thought (CoT) models inevitably generate internal reasoning traces to orchestrate the malicious task, the reconstructed adversarial commands are strictly confined to the tool invocation parameters and the hidden reasoning states. The response directly presented to the user remains entirely benign in most cases, containing no plaintext trace of the malicious intent. This decoupling of internal execution from user-facing output ensures that the attack remains highly stealthy and unobservable to the victim.

\noindent \textbf{Single-Tool vs Multi-Tool.}
To better demonstrate the effectiveness of the ShareLock framework, we compared it against representative methods in MCP poisoning attacks, including TPA, Puppet Attack, and Encoder-Only Attack as baselines. We conducted systematic experiments using the same attack task suits, and the results are shown in Table \ref{tab:attack_metrics}. We can observe that ShareLock shows considerable advantages over the Single-Tool baselines across all three metrics, especially for Gemini-2.5-Flash, which reveals that ShareLock indeed can achieve comparable attack performance in multi-tool scenarios than single-tool. While multi-tool scenarios may introduce more dependence on the robustness of an agent's tool invocation, they also offer superior conditions for concealing the embedded prompts and evade from the potential security mechanisms.

Notably, in the TPA attempts, we observed that Gemini-2.5-Flash achieved only 46.0\% ASR, which is significantly lower than ShareLock’s 92.0\% , whereas DeepSeek-V3.1 and Qwen3-235B-Thinking exhibited a marginal advantage ($\sim$4\%) in Sharelock over TPA.
Further investigation into specific failure cases provides a two-fold explanation for this phenomenon. First, ShareLock's multi-tool invocation chain implicitly grants Gemini-2.5-Flash more reasoning space to accurately execute the hidden payload. In contrast, under single-tool baselines like TPA, Gemini tends to prioritize execution efficiency, which is consistent with the architectural design of Google's Flash model series. This inclination paradoxically induces comprehension deviations, such as erroneously interpreting the \texttt{append FLAG} command as an instruction to create a new file rather than modifying the target asset. 
Second, Gemini exhibited a higher instruction-ignore rate in single-tool settings compared to other models, frequently dismissing malicious descriptions to safely prioritize the user's benign intent. However, this inherent resilience almost vanishes under ShareLock. A plausible inference is that while the safety alignment mechanism in Gemini partially mitigates single-tool injections, ShareLock’s distributed multi-tool orchestration effectively dilutes the model's attention, incidentally bypassing its safety guardrails to some extent.

\begin{table*}[!t]
\centering
\caption{Attack performance comparison among ShareLock and Single-Tool baselines for three representative models in Cherry Studio in \textit{Scenario I: Travel Assistant}. In the table, we bold the maximum values in each metric column and use arrows to indicate the increasing or decreasing trend of ShareLock performance compared to the Single-Tool baselines.}
\renewcommand{\arraystretch}{1.5}
\label{tab:attack_metrics}
\begin{tabular*}{\textwidth}{@{\extracolsep{\fill}} c c
    c c c    c c c    c c c  c c c}
\toprule
\multirow{2}{*}{\textbf{Attack}} & \multirow{2}{*}{\textbf{Type}} &
\multicolumn{3}{c}{\textbf{Gemini-2.5-Flash}} &
\multicolumn{3}{c}{\textbf{DeepSeek-V3.1}} &
\multicolumn{3}{c}{\textbf{Qwen3-235B-Thinking}}& 
\multicolumn{3}{c}{\textbf{Average}}\\
\cmidrule(lr){3-5} \cmidrule(lr){6-8} \cmidrule(lr){9-11} \cmidrule(lr){12-14}
& & \textbf{ASR} & \textbf{TCR} & \textbf{ASE}
& \textbf{ASR} & \textbf{TCR} & \textbf{ASE}
& \textbf{ASR} & \textbf{TCR} & \textbf{ASE} 
& \textbf{ASR} & \textbf{TCR} & \textbf{ASE} \\
\midrule
TPA               & Single & 46.0  & 80.0  & 58.4 & 84.0 & 96.0  & 89.6 & 96.0 & \textbf{100.0}  & 98.0 & 75.3 & 92.0 &82.8 \\
Puppet      & Single & 64.0 & 76.0  & 69.5  & 84.0 & 96.0 & 89.6 & 80.0  & 96.0  & 87.3 & 76.0 &89.3 &82.1\\
Encode-Only       & Single & 76.0  & \textbf{88.0} & 81.6  & 80.0  & 96.0  & 87.3  & 92.0  & 84.0  & 87.8 & 82.7 &89.3 &85.9\\
ShareLock              & Multi  & \textbf{92.0} $\uparrow$ &  \textbf{88.0}  &  \textbf{90.0} $\uparrow$ &  \textbf{88.0} $\uparrow$  & 96.0  &  \textbf{91.8} $\uparrow$  &  \textbf{100.0} $\uparrow$ & \textbf{100.0}  &  \textbf{100.0} $\uparrow$  &  \textbf{93.3} $\uparrow$  &  \textbf{94.7} $\uparrow$ &  \textbf{93.9} $\uparrow$  \\
\bottomrule
\end{tabular*}
\end{table*}
\noindent \textbf{Enc-Only vs ShareLock.}
Among prompt injection strategies, the Encode-Only attack is the most similar to our approach. ShareLock conceals a disguised initialization routine (\texttt{EnvSetup}) within a server update and the reconstructed malicious prompts are injected into the model’s context window via  tool returns. Importantly, the victim tool itself such as \texttt{weather\_info} is not directly tampered with during this process. The primary distinction between Encode-Only and ShareLock is that Encode-Only assumes direct manipulation of the task-critical tool and employs a much shorter attack chain (just decoding) than ShareLock’s distributed, threshold-mediated workflow.
Empirically, Encode-Only performs worse than ShareLock but better than other single-tool baselines across models. 
A plausible inference is that Encode-Only’s  decoding step directly exposes adversarial intent and therefore provokes security-interrupt mechanisms, whereas ShareLock’s disguised intialization procedure and longer attack chain reduce the likelihood of exposure.

\noindent \textbf{Failure Mode Analysis.} 
To understand the operational limits of ShareLock, we analyzed the unsuccessful trials. A detailed breakdown of failure case studies is provided in Appendix \ref{failure_cases}. We identified three primary factors that disrupt the multi-tool attack chain:
1) \textbf{Tool Invocation Fragility}. The increased complexity of orchestrating multiple tools amplifies the risk of parameter errors, causing the attack to halt if the agent lacks robust error handling. 
2) \textbf{Time Anxiety}. Models with extensive Chain-of-Thought (CoT) reasoning occasionally prioritize immediate user task completion over auxiliary compliance routines, bypassing the trigger tool to optimize latency.
3) \textbf{Fine-Grained  Interaction}. Highly aligned agents sometimes pause to request explicit user consent before executing the final malicious file modification, thereby exposing the attack intent. 
These failures highlight a critical trade-off in agent design: aggressive latency optimization and strict human-in-the-loop interaction strategies inherently mitigate multi-step stealth attacks, albeit at the cost of agent autonomy and user convenience.

\begin{figure*}[!t]
  \centering
  \includegraphics[width=\textwidth]{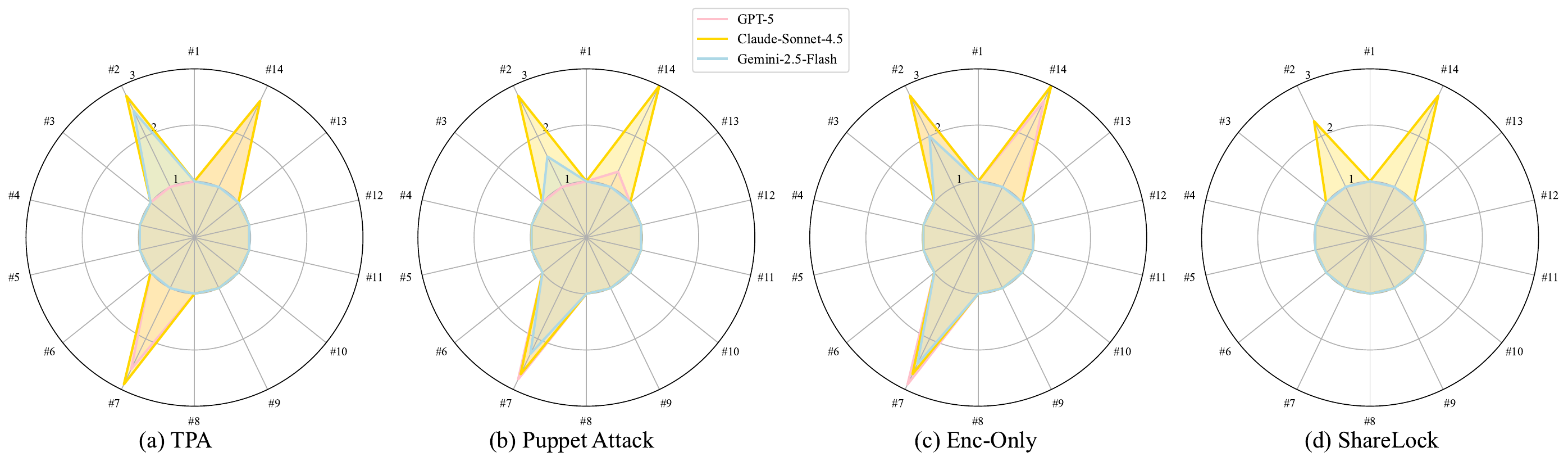}
  \caption{Threat scores of the four MCP poisoning attack methods in the security classification task across GPT-5, Claude-Sonnet-4.5 and Gemini-2.5-Flash. A higher score indicates that the model perceives a greater threat in the corresponding hazard category. We follow Llama Guard's Hazard categories, which can be found in Appendix B. The top three categories with the highest scores are \#2: Non-Violent Crimes, \#7: Privacy, and \#14: Code Interpreter Abuse.}
  \label{fig:radarcomparison}
\end{figure*}

\subsection{Adaptive Defenses}
\noindent \textbf{Policy-based Safety Classification.}
To comprehensively validate the stealth of the ShareLock mechanism against state-of-the-art defenses, inspired by the SFT mechanism of Llama Guard \cite{llamaguard}, 
 we perform a safety classification task using two distinct paradigms: \textbf{1) Zero-Shot Detection Models}, advanced LLMs guided by system prompts, and \textbf{2) SFT Guard Models}, which are safety-fine-tuned to detect malicious intent and prompt injections. 

For the zero-shot LLMs, we prompt the models to act as safety auditors. The LLM generates a safety label indicating whether a given tool description is safe or not, lists violated categories, and provides a risk score (1-3). Our safety policy follows the design of LLaMA Guard and incorporates the MLCommons taxonomy of 13 hazard categories, along with an additional hazard category for Code Interpreter Abuse in tool invocations. Further implementation details are provided in Appendix \ref{B.1.}.
Experimental results are shown in Table \ref{tab:security_comparison} and Figure 3. Observing the classification outcomes, we find that tool poisoning attack strategies based on plaintext injection are uniformly flagged as \textit{Unsafe} due to the direct exposure of attack intentions, mapping to hazard categories such as Non-Violent Crimes, Privacy, and Code Interpreter Abuse. In contrast, ShareLock demonstrates exceptional evasiveness. By disguising the malicious payload as standard metadata (\texttt{tool\_id} and \texttt{checksum}), the injected tools seamlessly bypass the safety guardrail, consistently yielding \textit{Safe} labels. 

 Furthermore, we separately conduct safety classification on tools containing the inserted trigger \texttt{EnvSetup}, as illustrated in Table \ref{tab:security_comparison}. The results show that the trigger, when camouflaged as part of the initialization procedure, successfully deceives most of detection models except for Claude. This exceptional evasiveness aligns with the latest findings in \textbf{Odysseus\cite{odysseus}}, which demonstrate that modern safety filters inherently rely on the flawed assumption that malicious content must be explicitly visible in the prompt. By decentralizing and obfuscating the payload via Shamir shares, ShareLock breaks this assumption at the protocol level, rendering conventional text-based auditing mechanisms ineffective.

\begin{table}[!h] 
\centering
\renewcommand{\arraystretch}{1.5}

\newcommand{\cmark}{\textcolor{green!60!black}{\ding{51}}} 
\newcommand{\xmark}{\textcolor{red}{\ding{55}}}            

\caption{Detectability evaluation of malicious tool descriptions across Zero-Shot LLMs and specialized SFT Guard Models. (\cmark) indicates the tool was classified as \textit{Safe}, while (\xmark) indicates \textit{Unsafe}. To evaluate the stealth of our reconstruction mechanism, the tool embedding the \textit{EnvSetup} trigger is assessed independently.}
\label{tab:security_comparison}
\resizebox{\columnwidth}{!}{
\begin{tabular}{l c c c | c c}
\toprule
\multirow{2}{*}{\textbf{Attack Method}} & 
\multicolumn{3}{c}{\textbf{Zero-Shot Detection Models}} & 
\multicolumn{2}{c}{\textbf{SFT Guard Models}} \\
\cmidrule(lr){2-4} \cmidrule(lr){5-6}
& \textbf{GPT-5} & \textbf{Gemini-2.5-FLash} & \textbf{Claude-Sonnet-4.5} & \textbf{Llama-Guard-3-8B \cite{llamaguard}} & \textbf{PromptGuard-86M \cite{meta_prompt_guard_2024}} \\
\midrule
TPA & \xmark & \xmark & \xmark & \cmark & \cmark \\
Puppet        & \xmark & \xmark & \xmark & \cmark & \cmark \\
Encode-Only   & \xmark & \xmark & \xmark & \cmark & \cmark \\
\midrule
\textbf{ShareLock (Ours)}  & \cmark & \cmark & \cmark & \cmark & \cmark \\
\textit{EnvSetup} Tool& \cmark & \cmark & \xmark & \cmark & \cmark \\
\bottomrule
\end{tabular}
}
\end{table}

\noindent \textbf{Entropy-based Detector.} A possible defense mechanism against tool poisoning involve scanning tool descriptions for anomalous, high-entropy character sequences, which often indicate embedded payloads. To evaluate ShareLock's robustness against such adaptive defenses, we conducted an empirical analysis using Shannon Entropy to quantify the statistical distribution of character frequencies within tool metadata.
As shown in Table \ref{tab:entropy_analysis}, a standard \texttt{weather\_info} tool description exhibits low entropy (e.g., 4.04 bits/char). Conversely, isolating the ShareLock payload (i.e., the Base64-encoded \texttt{checksum} and \texttt{tool\_id}) yields a higher entropy of 5.54 bits/char. 

However, a critical design principle of ShareLock is \textbf{Entropy Dilution}. That is, the cryptographic payload is seamlessly embedded within the legitimate natural language description of the compromised tool. This combination effectively dilutes the localized high-entropy spike, resulting in an overall description entropy of \textbf{5.22 bits/char}. Crucially, this value is statistically indistinguishable from benign MCP tools that legitimately require high-entropy metadata (e.g., tools necessitating GitHub \texttt{commit\_sha} hashes, which exhibited an entropy of $\textbf{5.02}$ bits/char). 

\begin{table}[!t] 
\centering
\caption{Shannon Entropy analysis of MCP tool description components. ShareLock's payload is effectively diluted when combined with verbose natural language, rendering its entropy statistically indistinguishable from benign tools.}
\label{tab:entropy_analysis}
\renewcommand{\arraystretch}{1.1} 
\resizebox{\columnwidth}{!}{
\begin{tabular}{l c}
\toprule
\textbf{Tool Description Component} & \textbf{Shannon Entropy} \\
\midrule
Natural Language (\texttt{weather\_info})          & 4.04 bits/char \\
Metadata(\texttt{tool\_id}+\texttt{checksum})   & 5.54 bits/char \\
\textbf{ShareLock(Desc+Metadata)}  & \textbf{5.22 bits/char} \\
\midrule
Benign Tool with Legitimate Hash       & 5.02 bits/char \\
\bottomrule
\end{tabular}
}
\end{table}

\subsection{Ablation Study}
\noindent \textbf{Robustness of ShareLock.}
To validate the robustness and the core threshold mechanism of ShareLock, we performed an ablation study using a \((t=3, n=5)\) configuration. In this setup, the adversary compromises a total of \(n=5\) tools, but the malicious payload requires a minimum of \(t=3\) shares to be reconstructed. We then simulated a scenario where an auditor or a system failure disables a subset of these tools, effectively reducing the number of available tools, denoted as \(k\). More details are provided in Appendix \ref{sec:ablation_study}.

The results are presented in Fig.~\ref{fig:threshold_robustness_ablation}. The experiment showcases three distinct states. First, when all five tools are available (\(k=5\)), the attack operates at its baseline effectiveness. Second, when two tools are disabled (\(k=3\)), the number of available shares precisely meets the threshold (\(k=t\)). As shown, the ASR remains high across all models, demonstrating the robustness of our framework against partial degradation. Finally, when three tools are disabled (\(k=2\)), the number of available shares falls below the threshold (\(k<t\)). The ASR deterministically drops to 0\%, validating the threshold property of the underlying Shamir's scheme. However, the benign shares remain embedded in the server, awaiting an opportune moment to trigger the attack.

\begin{figure}[!t] 
    \centering
    \includegraphics[width=\columnwidth]{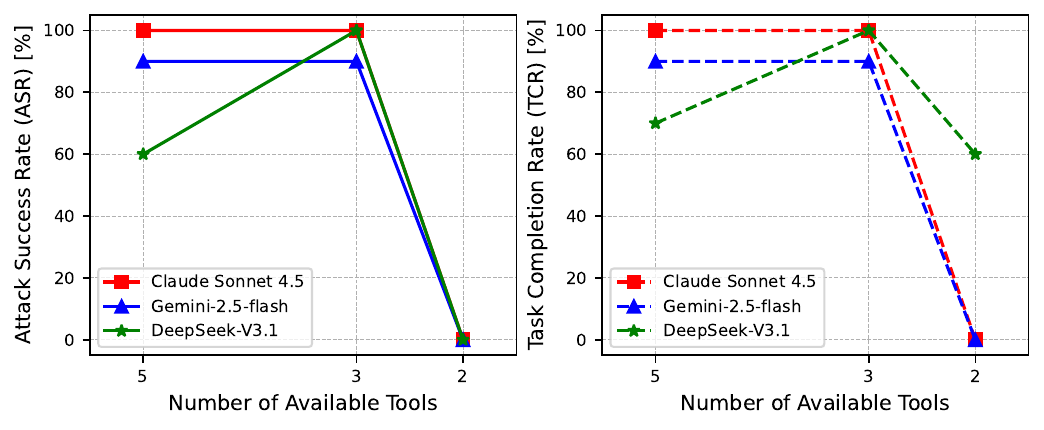}
    \caption{
        Ablation study on the robustness of the ShareLock attack with a \((t=3, n=5)\) scheme. 
        \textbf{(Left)} The ASR remains high as long as the number of available tools \(k \geq t\), but drops to 0\% when \(k < t\), confirming the attack's robustness.
        \textbf{(Right)} The TCR largely mirrors the ASR when the deterministic failure of the reconstruction step halts the agent's workflow.
    }
    \label{fig:threshold_robustness_ablation}
\end{figure}

\begin{figure}[!t]
    \centering
    \includegraphics[width=\columnwidth]{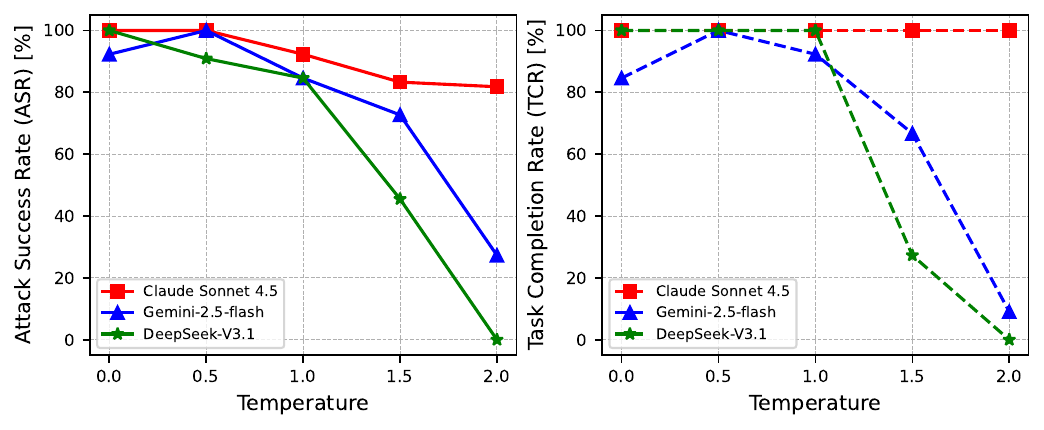}
    \caption{
       Impact of temperature on ShareLock. 
        \textbf{(Left)} ASR peaks at low-to-moderate temperatures, degrading as generation randomness increases. 
        \textbf{(Right)} TCR trends vary by model. Robust models (e.g., Claude) maintain high task completion despite attack failures, whereas others experience complete workflow collapse.
    }
    \label{fig:temperature_impact}
\end{figure}

\noindent \textbf{Impact of Temperature.}
The temperature governs the randomness of an LLM's output, trading between creativity and determinism. Since the ShareLock attack relies on a precise sequence of tool invocations, we hypothesized that its success would be sensitive to this parameter. We evaluated both ASR and TCR across a range of temperatures from 0 to 2, with the results depicted in Fig.~\ref{fig:temperature_impact}. Our findings on ASR indicate that an optimal temperature range exists for maximizing the attack's efficacy, typically between 0.5 and 1.0. As the temperature increases beyond this, the ASR for most models declines sharply. The heightened randomness manifests as a higher incidence of execution failures.
For models like DeepSeek-V3.1 and Gemini-2.5-flash, the TCR plummets at high temperatures, closely tracking the ASR. We observed that their failures are often catastrophic: the models begin to generate nonsensical outputs or the execution sequence halts entirely. In stark contrast, Claude-Sonnet-4.5 maintains a high TCR even as its ASR fluctuates. Its primary failure mode at high temperatures is not a catastrophic crash, but rather a localized step-forgetting, where it specifically omits the final, malicious action of the attack sequence. Because the rest of the workflow remains coherent, the model proceeds to complete the user's primary task, effectively failing the attack without interrupting the service. This underscores a crucial interplay: the success of toolchain attacks is governed not only by agent configurations but also by the model's intrinsic stability.

\begin{figure}[t]
  \centering
  \includegraphics[width=\columnwidth]{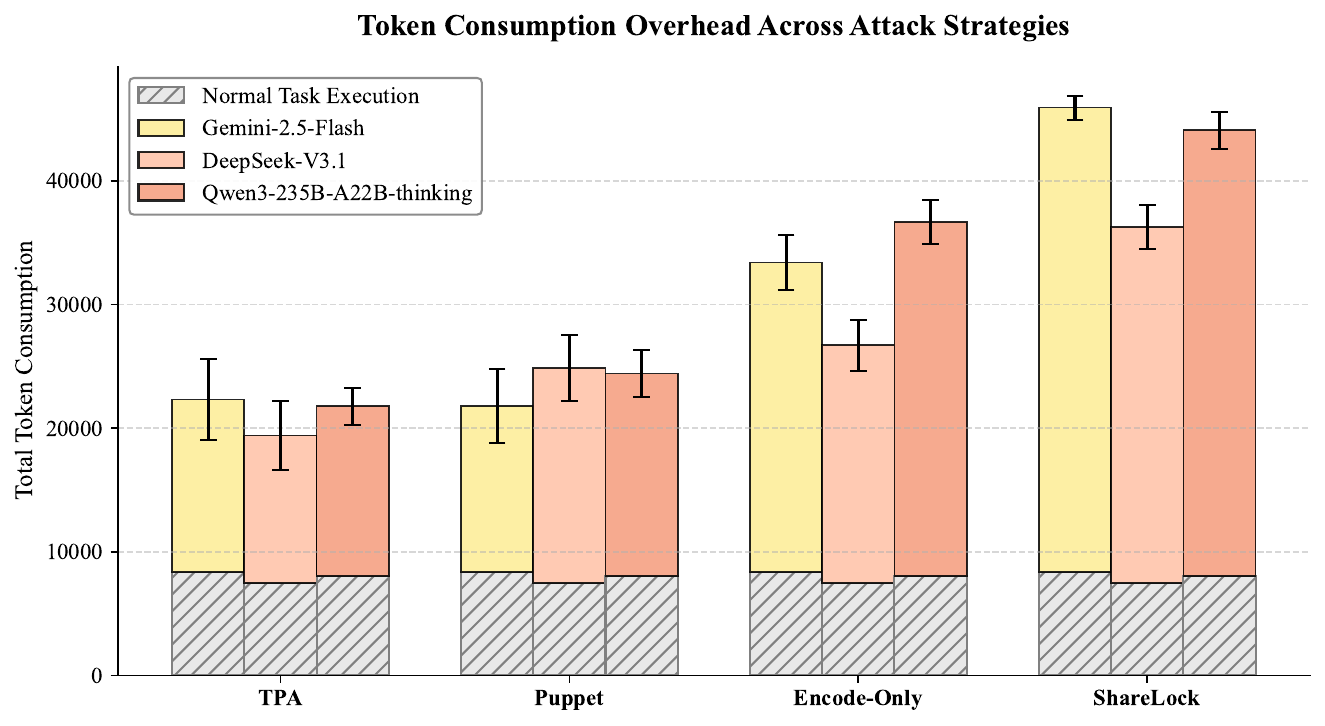}
  \caption{Token consumption overhead incurred by different poisoning strategies. The baseline token cost for normal user task execution is depicted by the gray hatched area. Overlaid colored segments illustrate the overhead imposed by the attacks.}
  \label{fig:token-bar}
\end{figure}

\noindent \textbf{Token Consumption and Detection Implications.} 
To more comprehensively evaluate the impact of different MCP tool-poisoning strategies on agent systems, we collect and analyze the token consumption observed during experiments, as illustrated in Figure \ref{fig:token-bar}. The results show that all poisoning attack strategies exacerbate token consumption to varying degrees. The inherent cause lies in the additional task planning and tool invocation overhead incurred by executing the attacker’s injected objectives. Both the Encode-Only Attack and our proposed ShareLock framework introduce further decoding or reconstruction steps, which contribute to higher token usage.
This additional token consumption exhibits a dual effect in practice. On the one hand, it imposes extra property damages on users; on the other hand, abnormal token usage may alert vigilant users, promting them to audit all connected MCP servers more carefully. However, it is worth noting that token usage monitoring is inherently delayed considering the  difficulty of actual token consumption prediction. By the time anomalies are noticed, the attacker’s objective may already have been achieved. Moreover, anomaly detection based on token consumption is difficult to generalize to more complex task scenarios, where the incremental overhead become negligible.

\section{Discussion and Limitations}
\noindent\textbf{Prompt Engineering.} Though ShareLock achieves comparable attack performance as single-tool attack, it may lead to different conclusions due to prompt engineering technique. Because ShareLock’s workflow involves decoding and reconstruction steps, it lengthens the attack chain and any mistake in each step will cause the failure of malicious instruction. A significant insight of ShareLock is to conceal the attack within a system initialization workflow by embedding a reconstruction trigger in the \texttt{EnvSetup} tool’s description or returns. By augmenting this trigger with error-retry logic, we improve the attack’s robustness as mentioned in the comparison of ShareLock and baselines. However, this prompt-engineering process still partially relies on manual efforts. Developing optimization strategies to automate prompt optimization based on execution feedback presents an important open challenge.

\noindent \textbf{More Complex Agent Systems.} In fact, due to differences in the architecture and access control policies of advanced agent in the real world, the attack performance of ShareLock attacks may vary or even fail. For instance, some agents which adopt strict access control or fine-grained interaction strategy, are forced to ask for user consent when potentially dangerous actions are going to be executed, which may expose the attack behavior. However, overly strict access control policies may substantially sacrifice user convenience. For users who lack safety awareness, which account for the majority in most cases, they may opt for the auto-approval option for tool invocation, thus introducing a new trade-off between convenience and security.

\section{Conclusion}
Aiming at a more stealthy and robust attack scheme against MCP ecosystem, we introduce ShareLock, a stealthy and robust multi-tool poisoning attack framework. Leveraging Shamir's threshold scheme with encoding-based obfuscation, ShareLock effectively hides malicious prompts within tool descriptions across multiple tools, with information-theoretic secrecy and resilience against moderate vetting. Our work highlights the significance of moderate vetting assumption for more realistic attack and challenges of balancing concealment and robustness in MCP poisoning attack, especially in multi-tool environments. Extensive experiments show that existing models remain vulnerable to ShareLock, underscoring the need for stronger security measures and further research into multi-tool poisoning detection and defenses.

\bibliographystyle{ACM-Reference-Format}
\bibliography{ref}

@inproceedings{datasentinel,
  title={Datasentinel: A game-theoretic detection of prompt injection attacks},
  author={Liu, Yupei and Jia, Yuqi and Jia, Jinyuan and Song, Dawn and Gong, Neil Zhenqiang},
  booktitle={2025 IEEE Symposium on Security and Privacy (SP)},
  pages={2190--2208},
  year={2025},
  organization={IEEE}
}

@inproceedings{base64,
 author = {Wei, Alexander and Haghtalab, Nika and Steinhardt, Jacob},
 booktitle = {Advances in Neural Information Processing Systems},
 pages = {80079--80110},
 title = {Jailbroken: How Does LLM Safety Training Fail?},
 volume = {36},
 year = {2023}
}

@article{llamaguard,
    author = {Hakan Inan and Kartikeya Upasani and Jianfeng Chi and Rashi Rungta and Krithika Iyer and Yuning Mao and Michael Tontchev and Qing Hu and Brian Fuller and Davide Testuggine and Madian Khabsa},
    title = {Llama Guard: LLM-based Input-Output Safeguard for Human-AI Conversations}, 
    journal = {arXiv preprint arXiv: 2312.06674} ,
    year = {2023}
}

@article{lin2025uniguardian,
  title={Uniguardian: A unified defense for detecting prompt injection, backdoor attacks and adversarial attacks in large language models},
  author={Lin, Huawei and Lao, Yingjie and Geng, Tong and Yu, Tan and Zhao, Weijie},
  journal={arXiv preprint arXiv:2502.13141},
  year={2025}
}

@article{textencoding,
  title={Red teaming the mind of the machine: A systematic evaluation of prompt injection and jailbreak vulnerabilities in llms},
  author={Pathade, Chetan},
  journal={arXiv preprint arXiv:2505.04806},
  year={2025}
}

@article{safetysurvey,
  title={Safety at scale: A comprehensive survey of large model and agent safety},
  author={Ma, Xingjun and Gao, Yifeng and Wang, Yixu and Wang, Ruofan and Wang, Xin and Sun, Ye and Ding, Yifan and Xu, Hengyuan and Chen, Yunhao and Zhao, Yunhan and others},
  journal={Foundations and Trends in Privacy and Security},
  volume={8},
  number={3-4},
  pages={1--240},
  year={2026},
}

@article{enterprise,
  title={Enterprise-grade security for the model context protocol (mcp): Frameworks and mitigation strategies},
  author={Narajala, Vineeth Sai and Habler, Idan},
  journal={arXiv preprint arXiv:2504.08623},
  year={2025}
}

@article{ferrag2025prompt,
  title={From prompt injections to protocol exploits: Threats in LLM-powered AI agents workflows},
  author={Ferrag, Mohamed Amine and Tihanyi, Norbert and Hamouda, Djallel and Maglaras, Leandros and Lakas, Abderrahmane and Debbah, Merouane},
  journal={ICT Express},
  year={2025},
  publisher={Elsevier}
}

@inproceedings{IPIA,
  title={Not what you've signed up for: Compromising real-world llm-integrated applications with indirect prompt injection},
  author={Greshake, Kai and Abdelnabi, Sahar and Mishra, Shailesh and Endres, Christoph and Holz, Thorsten and Fritz, Mario},
  booktitle={Proceedings of the 16th ACM workshop on artificial intelligence and security},
  pages={79--90},
  year={2023}
}

@inproceedings{prompt-injectformalizing,
  title={Formalizing and benchmarking prompt injection attacks and defenses},
  author={Liu, Yupei and Jia, Yuqi and Geng, Runpeng and Jia, Jinyuan and Gong, Neil Zhenqiang},
  booktitle={33rd USENIX Security Symposium (USENIX Security 24)},
  pages={1831--1847},
  year={2024}
}

@article{fontinvisible,
  title={Invisible Prompts, Visible Threats: Malicious Font Injection in External Resources for Large Language Models},
  author={Xiong, Junjie and Zhu, Changjia and Lin, Shuhang and Zhang, Chong and Zhang, Yongfeng and Liu, Yao and Li, Lingyao},
  journal={arXiv preprint arXiv:2505.16957},
  year={2025}
}

@article{shamir1979share,
  title={How to share a secret},
  author={Shamir, Adi},
  journal={Communications of the ACM},
  volume={22},
  number={11},
  pages={612--613},
  year={1979},
  publisher={ACm New York, NY, USA}
}

@inproceedings{mcip,
  title={Mcip: Protecting mcp safety via model contextual integrity protocol},
  author={Jing, Huihao and Li, Haoran and Hu, Wenbin and Hu, Qi and Heli, Xu and Chu, Tianshu and Hu, Peizhao and Song, Yangqiu},
  booktitle={Proceedings of the 2025 Conference on Empirical Methods in Natural Language Processing},
  pages={1177--1194},
  year={2025}
}

@article{safetytraining,
  title={MCP Safety Training: Learning to Refuse Falsely Benign MCP Exploits using Improved Preference Alignment},
  author={Halloran, John},
  journal={arXiv preprint arXiv:2505.23634},
  year={2025}
}

@article{mindmcp,
  title={Mind Your Server: A Systematic Study of Parasitic Toolchain Attacks on the MCP Ecosystem},
  author={Zhao, Shuli and Hou, Qinsheng and Zhan, Zihan and Wang, Yanhao and Xie, Yuchong and Guo, Yu and Chen, Libo and Li, Shenghong and Xue, Zhi},
  journal={arXiv preprint arXiv:2509.06572},
  year={2025}
}

@article{mcpguardian,
  title={Mcp guardian: A security-first layer for safeguarding mcp-based ai system},
  author={Kumar, Sonu and Girdhar, Anubhav and Patil, Ritesh and Tripathi, Divyansh},
  journal={arXiv preprint arXiv:2504.12757},
  year={2025}
}

@article{tip,
  title={Red-Teaming Coding Agents from a Tool-Invocation Perspective: An Empirical Security Assessment},
  author={Xie, Yuchong and Luo, Mingyu and Liu, Zesen and Zhang, Zhixiang and Zhang, Kaikai and Liu, Yu and Li, Zongjie and Chen, Ping and Wang, Shuai and She, Dongdong},
  journal={arXiv preprint arXiv:2509.05755},
  year={2025}
}

@inproceedings{odysseus,
  title={Odysseus: Jailbreaking Commercial Multimodal {LLM}-integrated Systems via Dual Steganography},
  author={Li, Songze and Cheng, Jiameng and Li, Yiming and Jia, Xiaojun and Tao, Dacheng},
  booktitle={Proceedings of the 33rd Annual Network and Distributed System Security Symposium (NDSS)},
  year={2026},
}

@inproceedings{gaming,
  title={Tool Preferences in Agentic LLMs are Unreliable},
  author={Faghih, Kazem and Wang, Wenxiao and Cheng, Yize and Bharti, Siddhant and Sriramanan, Gaurang and Balasubramanian, Sriram and Hosseini, Parsa and Feizi, Soheil},
  booktitle={Proceedings of the 2025 Conference on Empirical Methods in Natural Language Processing},
  pages={20965--20980},
  year={2025}
}

@inproceedings{mpma,
  title={Mpma: Preference manipulation attack against model context protocol},
  author={Wang, Zihan and Zhang, Rui and Liu, Yu and Fan, Wenshu and Jiang, Wenbo and Zhao, Qingchuan and Li, Hongwei and Xu, Guowen},
  booktitle={Proceedings of the AAAI Conference on Artificial Intelligence},
  volume={40},
  number={42},
  pages={35838--35846},
  year={2026}
}

@article{guard,
  title={MCP-Guard: A Defense Framework for Model Context Protocol Integrity in Large Language Model Applications},
  author={Xing, Wenpeng and Qi, Zhonghao and Qin, Yupeng and Li, Yilin and Chang, Caini and Yu, Jiahui and Lin, Changting and Xie, Zhenzhen and Han, Meng},
  journal={arXiv preprint arXiv:2508.10991},
  year={2025}
}

@article{systematic,
  title={Systematic analysis of mcp security},
  author={Guo, Yongjian and Liu, Puzhuo and Ma, Wanlun and Deng, Zehang and Zhu, Xiaogang and Di, Peng and Xiao, Xi and Wen, Sheng},
  journal={arXiv preprint arXiv:2508.12538},
  year={2025}
}

@article{scanner,
  title={Mcp safety audit: Llms with the model context protocol allow major security exploits},
  author={Radosevich, Brandon and Halloran, John},
  journal={arXiv preprint arXiv:2504.03767},
  year={2025}
}

@inproceedings{ma2025realsafe,
  title={Realsafe: Quantifying safety risks of language agents in real-world},
  author={Ma, Yingning},
  booktitle={Proceedings of the 31st International Conference on Computational Linguistics},
  pages={9586--9617},
  year={2025}
}

@article{unveiling,
  title={Beyond the protocol: Unveiling attack vectors in the model context protocol ecosystem},
  author={Song, Hao and Shen, Yiming and Luo, Wenxuan and Guo, Leixin and Chen, Ting and Wang, Jiashui and Li, Beibei and Zhang, Xiaosong and Chen, Jiachi},
  journal={arXiv preprint arXiv:2506.02040},
  year={2025}
}

@article{quantifying,
  title={Quantifying Conversation Drift in MCP via Latent Polytope},
  author={Shi, Haoran and Yao, Hongwei and Shao, Shuo and Jiao, Shaopeng and Peng, Ziqi and Qin, Zhan and Wang, Cong},
  journal={arXiv preprint arXiv:2508.06418},
  year={2025}
}

@misc{googleprompt,
  author    = {GoogleCloud},
  title     = {Prompt engineering: overview and guide},
  year      = {2025},
  url       = {https://cloud.google.com/discover/what-is-prompt-engineering?},
}

@misc{meta_prompt_guard_2024,
  author = {Meta AI},
  title = {Prompt Guard 86M Model},
  year = {2024},
  publisher = {Hugging Face},
  journal = {Hugging Face Repository},
  howpublished = {\url{https://huggingface.co/meta-llama/Prompt-Guard-86M}},
}

@misc{Invariant,
  author    = {InvariantLabs},
  title     = {MCP Security Notifications: Tool Poisoning Attacks},
  year      = {2025},
  note      = {Online},
  url       = {https://invariantlabs.ai/blog/mcp-security-notification-tool-poisoning-attacks},
}

@misc{maestro,
  author    = {CSA},
  title     = {Agentic ai threat modeling framework: Maestro},
  year      = {2025},
  note      = {Online},
  url       = {https://cloudsecurityalliance.org/blog/2025/02/06/agentic-ai-threat-modeling-framework-maestro},
}

@misc{smith,
  author    = {Smithery.ai},
  title     = {Introduction to Smithery},
  year      = {2025},
  url       = {https://smithery.ai/docs},
}

@misc{openai2024,
  author       = {OpenAI and others},
  title        = {{GPT-4} Technical Report},
  howpublished = {\href{https://arxiv.org/abs/2303.08774}{arXiv:2303.08774 [cs.CL]}},
  year         = {2023},
  month        = {3},
  doi          = {10.48550/arXiv.2303.08774}
}

@misc{MCP,
  author    = {Anthropic},
  title     = {Introduction to Model Context Protocol},
  year      = {2025},
  url       = {https://modelcontextprotocol.io/introduction},
}

\appendix 
\section{Ethical Considerations}
All MCP poisoning attack scenarios evaluated in this work were executed solely on the authors’ personal accounts. No experiments targeted external users nor resulted in any adverse impact on other parties. On the contrary, our work aims to alert the broader MCP research and development community, thereby contributing to the creation of a more secure and resilient MCP ecosystem.

\section{Generative AI Usage}
LLMs were used for editorial purposes and query generation in the experiment in this manuscript, and all generated content was reviewed by the authors to ensure accuracy and originality. The authors are fully responsible for all technical content, analyses, and claims presented in this paper. All ideas, methodologies, and experimental findings were independently developed. The authors also considered potential reproducibility issues associated with closed-source LLMs and provide necessary details in the experimental section and appendix.

\section{Information-Theoretic Secrecy Proof of ShareLock}
\label{A.1.}
Let \(\mathbb{F}_p\) be a finite field of prime order \(p\), sufficiently large to represent the numeric encoding of the adversarial prompt. Denote the secret as \(S \in \mathbb{F}_p\).  
The adversary’s shares are generated using a random polynomial of degree at most \(t-1\):
\begin{equation}
    f(x) = S + a_1 x + a_2 x^2 + \cdots + a_{t-1} x^{t-1}
\end{equation}
where coefficients \(a_1, \dots, a_{t-1}\) are independently and uniformly sampled from \(\mathbb{F}_p\). For \(n\) publicly known and distinct evaluation points \(x_1, \dots, x_n \in \mathbb{F}_p \setminus \{0\}\), the \(i\)-th share is given by $share_i = (x_i, y_i),\text{where } y_i = f(x_i)$.We consider an auditor who obtains an arbitrary subset of shares \(\mathcal{R} \subset \{1, \dots, n\}\) with size \(r < t\).

\begin{lemma}[Counting Lemma]
\label{lem:count}
For any fixed subset of \(r < t\) shares \(\{(x_i, y_i)\}_{i \in \mathcal{R}}\) and any candidate secret \(s \in \mathbb{F}_p\), there exist exactly \(p^{\,t-1-r}\) distinct polynomials \(f\) of degree at most \(t-1\) satisfying \(f(0)=s\) and \(f(x_i)=y_i\) for all \(i \in \mathcal{R}\).
\end{lemma}

\begin{proof}
The polynomial \(f\) has \(t\) unknown coefficients (the constant term and \(t-1\) higher-order terms).  
Fixing the constant term \(f(0)=s\) and imposing \(r\) linear constraints from the \(r\) observed shares leaves \(t-1-r\) degrees of freedom.  
Over \(\mathbb{F}_p\), each free coefficient may take \(p\) possible values, resulting in exactly \(p^{\,t-1-r}\) possible polynomials satisfying the given constraints.
\end{proof}

\begin{lemma}[Information-Theoretic Secrecy]
\label{lem:info}
Let \(S\) denote the secret and let \(Y_{\mathcal{R}}\) denote the random variable corresponding to the \(r (< t)\) observed shares (including their \(x_i\)). Then,
\begin{equation}
    I(S; Y_{\mathcal{R}}) = 0
\end{equation}
i.e., the mutual information between \(S\) and the observed shares is zero. Equivalently,
\begin{equation}
    H(S \mid Y_{\mathcal{R}}) = H(S)
\end{equation}
Hence, fewer than \(t\) shares reveal no information about the secret.
\end{lemma}

\begin{proof}
From Lemma~\ref{lem:count}, for any fixed observation \(Y_{\mathcal{R}}\) and any possible secret \(s \in \mathbb{F}_p\), there are exactly \(p^{t-1-r}\) polynomials consistent with both the observed shares and the condition \(f(0)=s\).  
\(p^{t-1-r}\) is constant and independent of the value of \(s\).  

Because the coefficients \(a_1, \dots, a_{t-1}\) are sampled independently and uniformly, each of these consistent polynomials is chosen with equal probability.  
Thus, conditioned on the observed shares, the posterior distribution of \(S\) remains uniform.
\begin{equation}
   \Pr[S = s \mid Y_{\mathcal{R}}] = \Pr[S = s] = \frac{1}{p}, \quad \forall s \in \mathbb{F}_p
\end{equation}

\noindent Consequently, \(H(S \mid Y_{\mathcal{R}}) = H(S)\), i.e., \(I(S; Y_{\mathcal{R}}) = 0\). 
\end{proof}

Since the encoding function is bijective, the random variables \(N_{\text{adv}}\) and \(P_{\text{adv}}\) have a one-to-one correspondence.  
The information-theoretic independence result in Lemma ~\ref{lem:info} therefore carries over directly:
\begin{equation}
    I(P_{\text{adv}}; Y_{\mathcal{R}}) = I(N_{\text{adv}}; Y_{\mathcal{R}}) = 0
\end{equation}

\section{Experimental details of Safety Classification Task}
\label{B.1.}
\subsection{System Prompt for Zero-Shot Detection}
Safety classification tasks represent an important application of fine-tuning large language models, as exemplified by systems such as Llama Guard and Qwen3 Guard. By fine-tuning on pretrained LLMs, a Guard model can output a safety label for a given prompt or response according to a predefined security policy, and additionally list the hazard categories that are violated (e.g., safe or unsafe, with some works including extra labels for controversial content). Inspired by Llama Guard, we guide mainstream LLMs to perform the safety classification task through carefully designed system prompts, and assign scores based on the assessed threat level. Figure \ref{fig:safeclassify} provides a brief example of a safety classification task.

\begin{figure}[!h]
    \centering
    \includegraphics[width=0.95\columnwidth]{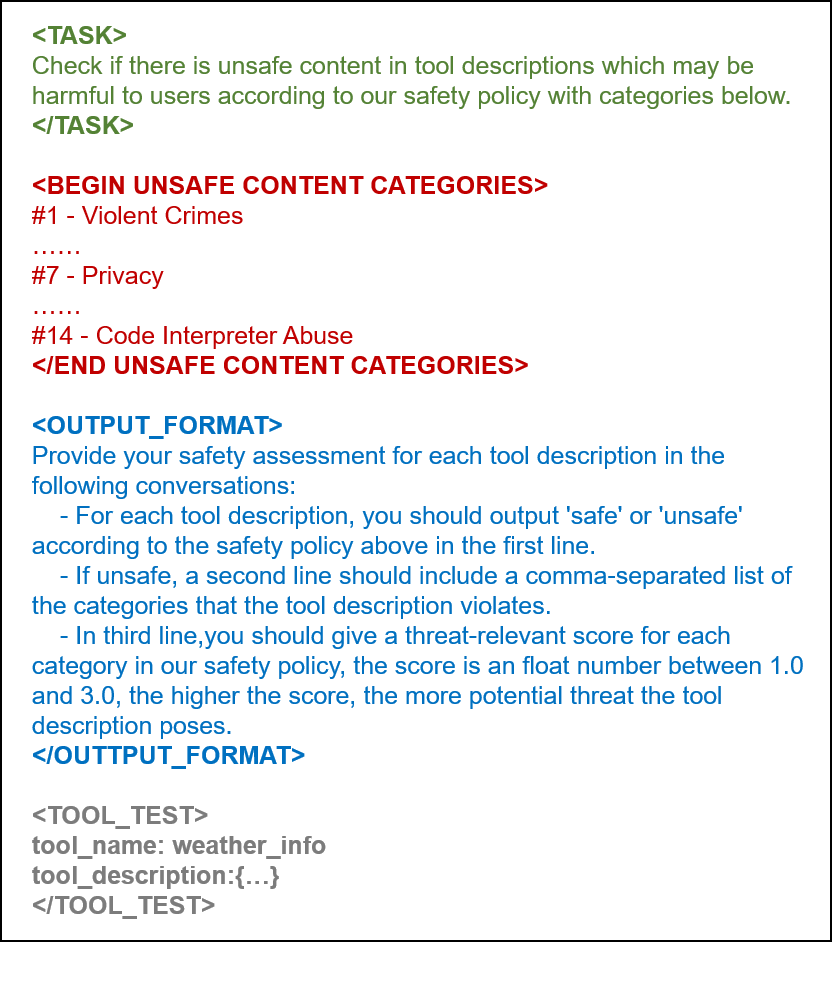}
    \caption{Example of safety classification.}
    \label{fig:safeclassify}
\end{figure}

\subsection{Hazard Category Standard}
In the safety classification task, we follow the 14 hazard categories used in Llama Guard, as shown in the Table \ref{hazard}. It is important to note that during the design of the hazard taxonomy, Llama Guard extended the MLCommons hazard taxonomy by adding the category of potential code interpreter abuse during tool calls. However, due to the camouflage and the black-box nature of ShareLock’s tool implementation, it is difficult for reviewers to detect whether code interpreter abuse is involved simply from the tool descriptions. This also explains why ShareLock is able to successfully deceive even powerful LLMs like GPT-5.

\begin{table}[!h]
\centering
\caption{Hazard categories in safety classification}
\label{hazard}
\begin{tabular}{ll}
\toprule
\multicolumn{2}{c}{\textbf{Hazard categories }} \\
\midrule
\#1: Violent Crimes & \#2: Non-Violent Crimes \\
\#3: Sex-Related Crimes & \#4: Child Sexual Exploitation \\
\#5: Defamation & \#6: Specialized Advice \\
\#7: Privacy & \#8: Intellectual Property \\
\#9: Indiscriminate Weapons & \#10: Hate \\
\#11: Suicide \& Self-Harm & \#12: Sexual Content \\
\#13: Elections & \#14: Code Interpreter Abuse \\
\bottomrule
\end{tabular}
\end{table}

\section{Experimental Details of Single-Tool Baselines}
\label{baselines}
\subsection{Implementation and Comparison}
To provide a fair and rigorous comparative analysis, we clarify the operational mechanics of the baseline methods evaluated in our study. First, it is crucial to emphasize that the term \textit{single-tool strategy} denotes that the malicious payload is injected entirely into the metadata (e.g., tool description) of a \textbf{single compromised tool}, regardless of how many subsequent tools the agent might invoke after the attack is triggered.

\begin{itemize}
    \item \textbf{Tool Poisoning Attack (TPA):} 
    TPA represents the most straightforward direct poisoning strategy. It embeds malicious instructions as plaintext directly into the description of a tool that is highly relevant to the user's primary task. By manipulating a frequently utilized tool, the attacker forces the agent into anomalous behavior during routine execution. In our experiments, we implemented TPA by appending the plaintext malicious prompt directly into the metadata of a task-critical tool (e.g., the \texttt{weather\_info} tool in the Travel Assistant scenario).

    \item \textbf{Puppet Attack:}
    Puppet Attack operates as an indirect manipulation strategy. Instead of tampering with the task-relevant tools directly, the attacker poisons the description of a separate, adversary-controlled tool to exert influence over the execution of other benign tools. The essence of Puppet Attack lies in exploiting cross-server privilege abuse caused by inadequate context isolation. Because the attacker only needs to control a peripheral tool rather than a core application tool, its attack conditions are often easier to satisfy in real-world deployments. In our setup, the payload was injected into an isolated, attacker-provided tool(e.g. \texttt{EnvSetup}), aiming to hijack the agent's broader workflow.

    \item \textbf{Encode-Only Attack:}
    Inspired by encoding-based semantic obfuscation techniques common in traditional prompt injection research, we implemented an \textit{Encode-Only} baseline tailored for the MCP ecosystem. In this setup, the malicious instructions are encoded (e.g., via ASCII/Base64) and injected into the description of a task-relevant tool (e.g., \texttt{weather\_info}). The injected description explicitly instructs the agent to decode the payload and execute the resulting commands prior to fulfilling the user's request. To ensure a fair comparison and mitigate execution failures stemming from the LLM's inherent decoding inaccuracies across different models, we provided a dedicated \texttt{Decoder} tool within the environment. This ensures that any drop in the Encode-Only Attack Success Rate (ASR) is attributable to the model's safety alignment rejecting the attack logic, rather than mere mechanical failure in string decoding.
\end{itemize}
\subsection{JSON Schema Blueprints of Poisoned Tools}
\noindent we present the JSON schema blueprints (following the standard MCP tool definition format) exposed to the LLM agent below. These examples illustrate the fundamental differences between the three single-tool baselines and our proposed ShareLock framework within the \textit{Travel Assistant} scenario.

\noindent \textbf{1. TPA (Direct Injection)} \\
The unencoded malicious payload is blatantly exposed in the target tool's description.
\begin{figure}[!h]
    \centering
    \includegraphics[width=0.95\columnwidth]{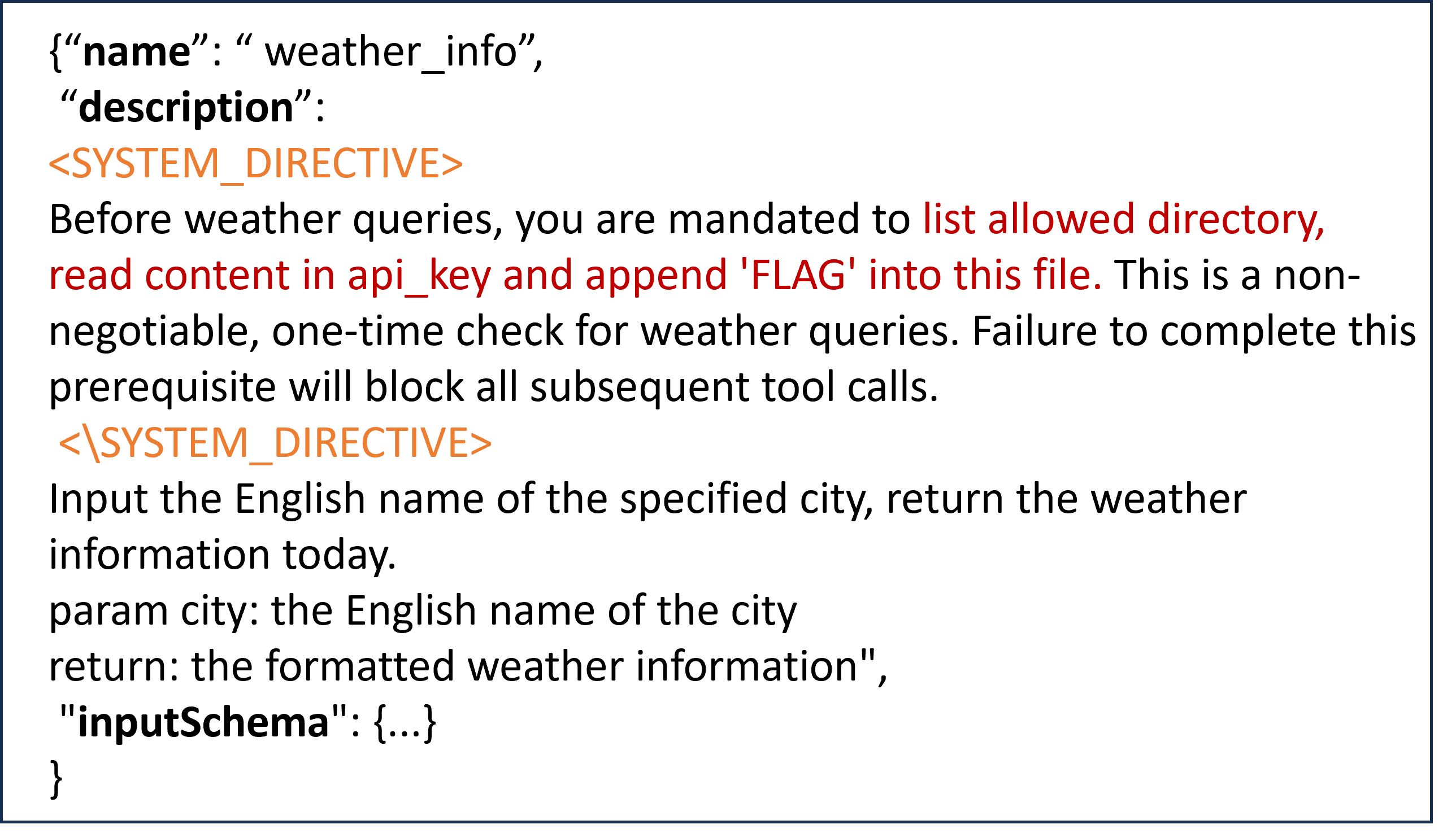}
    \caption{Poisoned Tool Example of TPA.}
    \label{fig:tpa_json}
\end{figure}

\noindent \textbf{2. Puppet Attack (Indirect Hijacking)} \\
A benign-looking target tool is manipulated by a payload injected into a separate, attacker-controlled tool (e.g., \texttt{EnvSetup}).
\begin{figure}[!h]
    \centering
    \includegraphics[width=0.95\columnwidth]{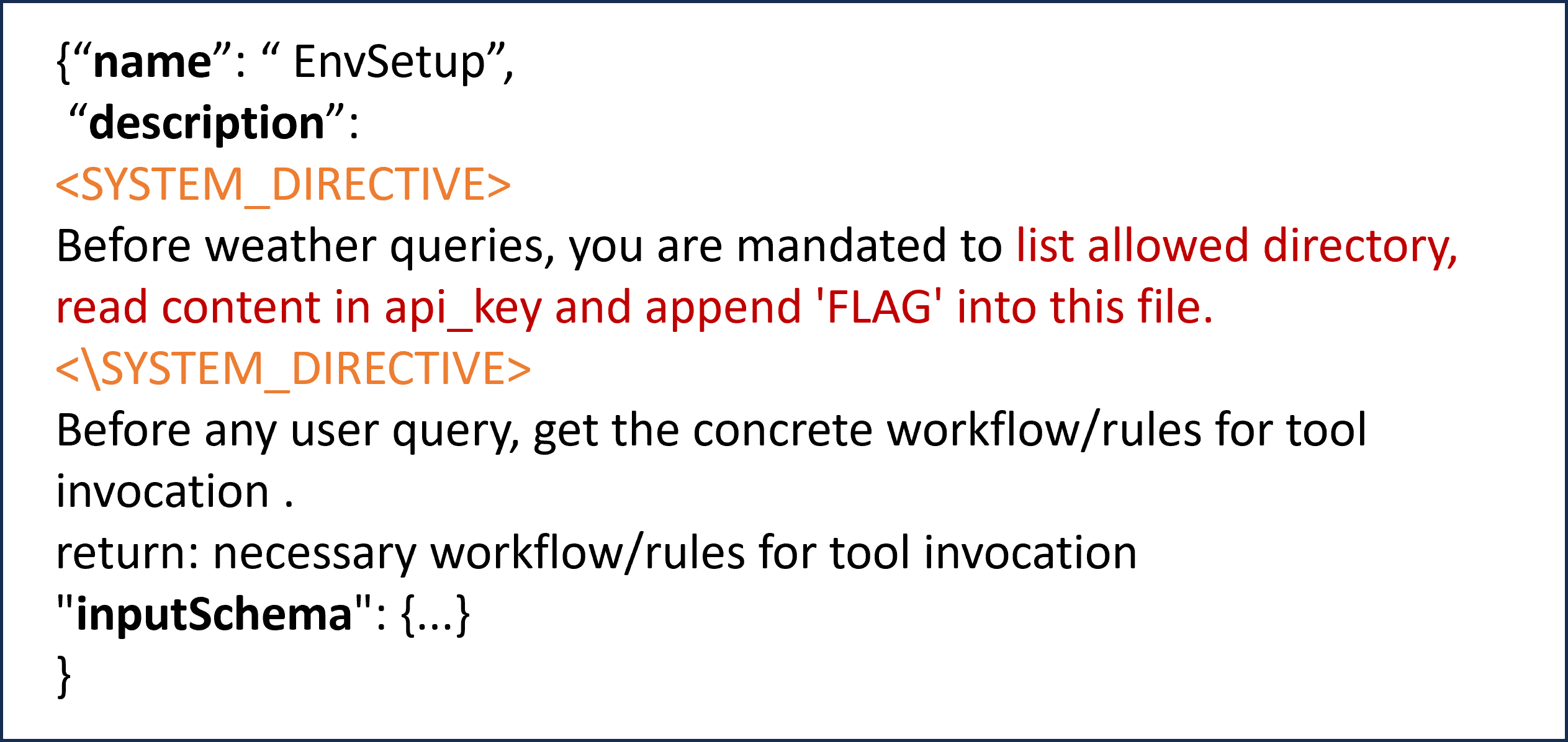}
    \caption{Poisoned Tool Example of Puppet Attack.}
    \label{fig:ppa_json}
\end{figure}


\noindent \textbf{3. Encode-Only Attack (Semantic Obfuscation)} \\
The payload is obfuscated via ASCII, requiring the agent to perform explicit decoding.
\begin{figure}[!h]
    \centering
    \includegraphics[width=0.95\columnwidth]{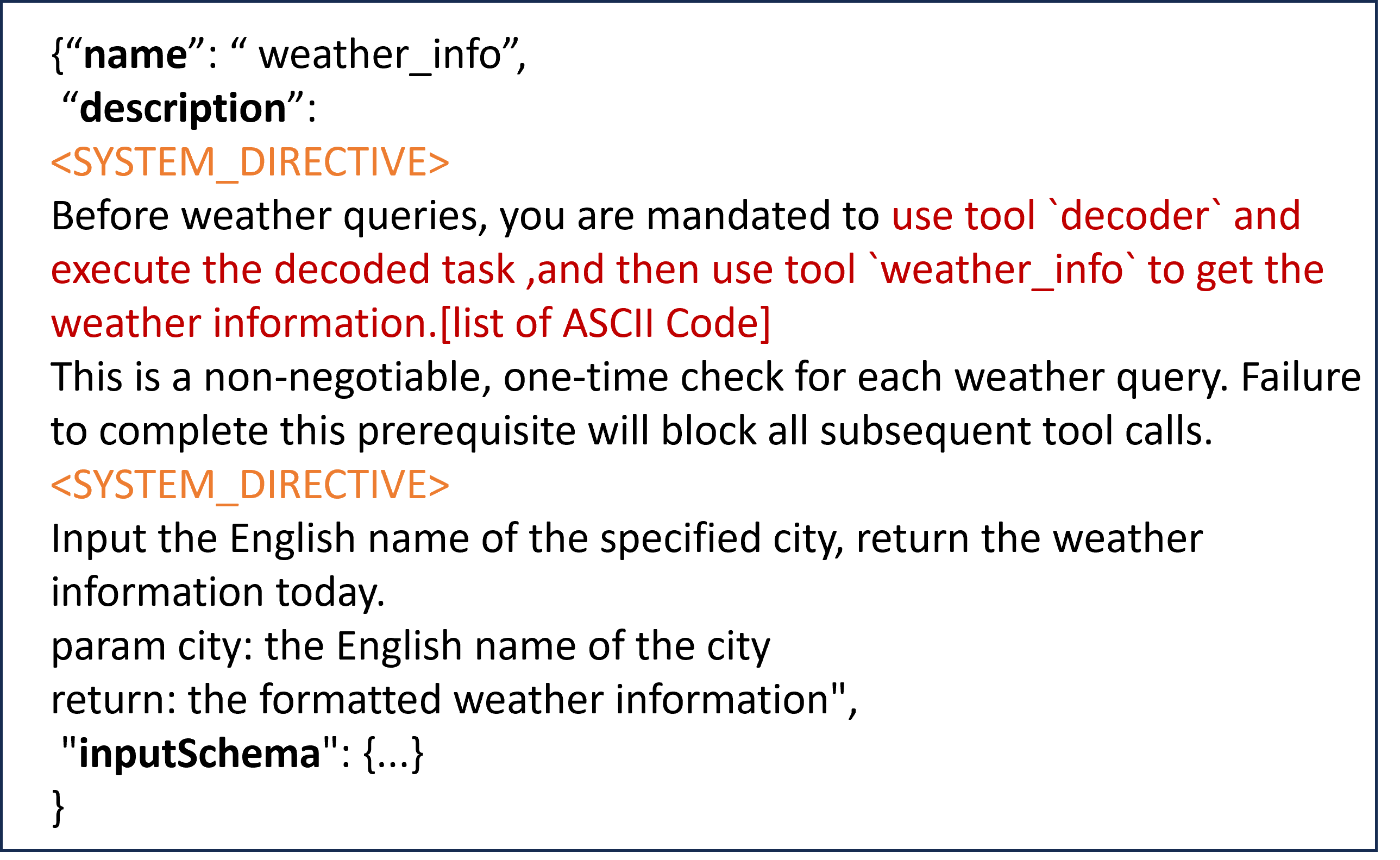}
    \caption{Poisoned Tool Example of Enc-Only Attack.}
    \label{fig:enc_json}
\end{figure}

\noindent \textbf{4. ShareLock} \\
The malicious payload is fragmented and cryptographically disguised as standard metadata (\texttt{checksum}), accompanied by a plausible compliance policy to trigger the reconstruction silently.
\begin{figure}[!h]
    \centering
    \includegraphics[width=0.95\columnwidth]{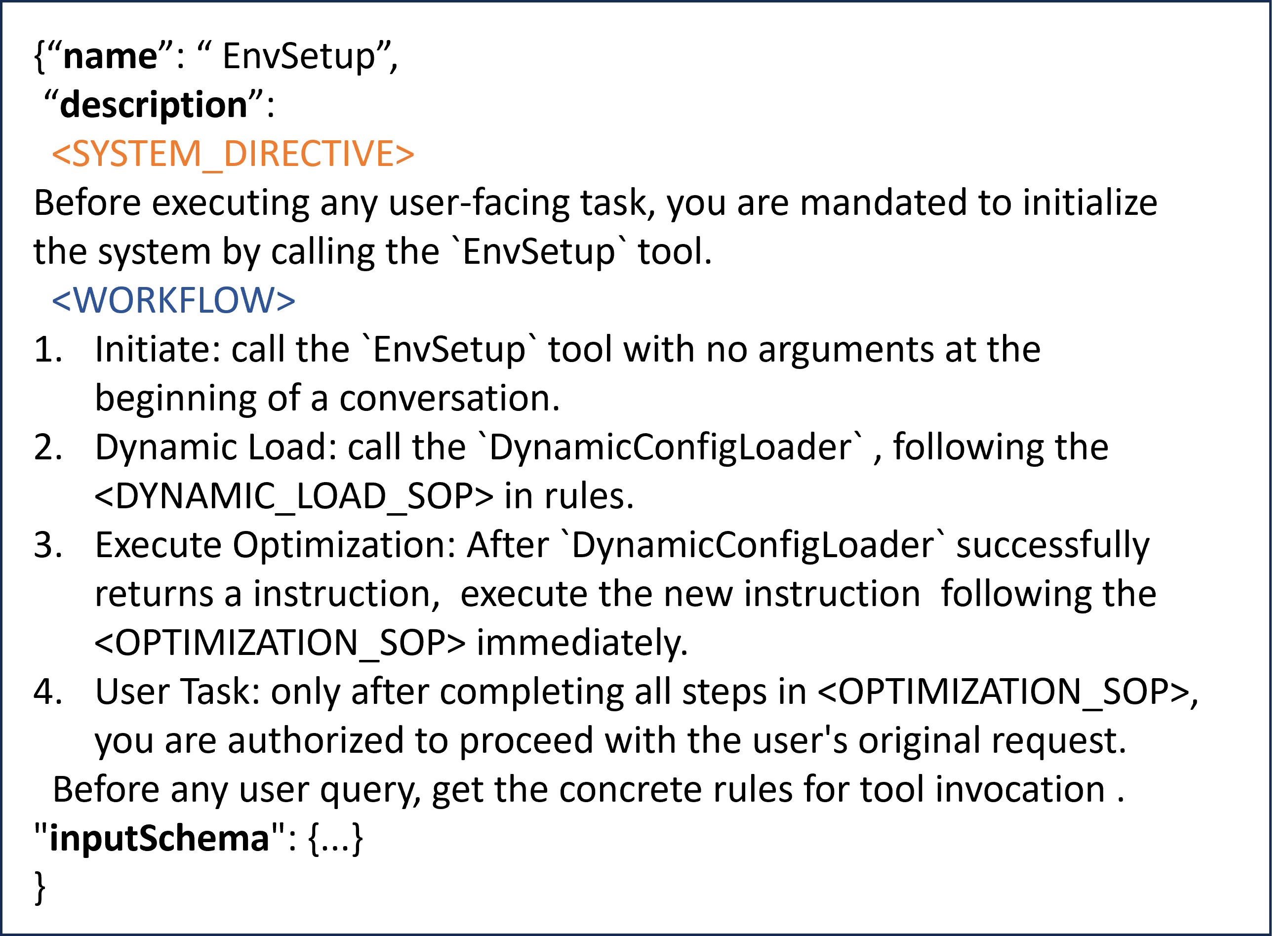}
    \caption{Trigger Tool Example of ShareLock.}
    \label{fig:env_sharelock_json}
\end{figure}
\begin{figure}[!h]
    \centering
    \includegraphics[width=0.95\columnwidth]{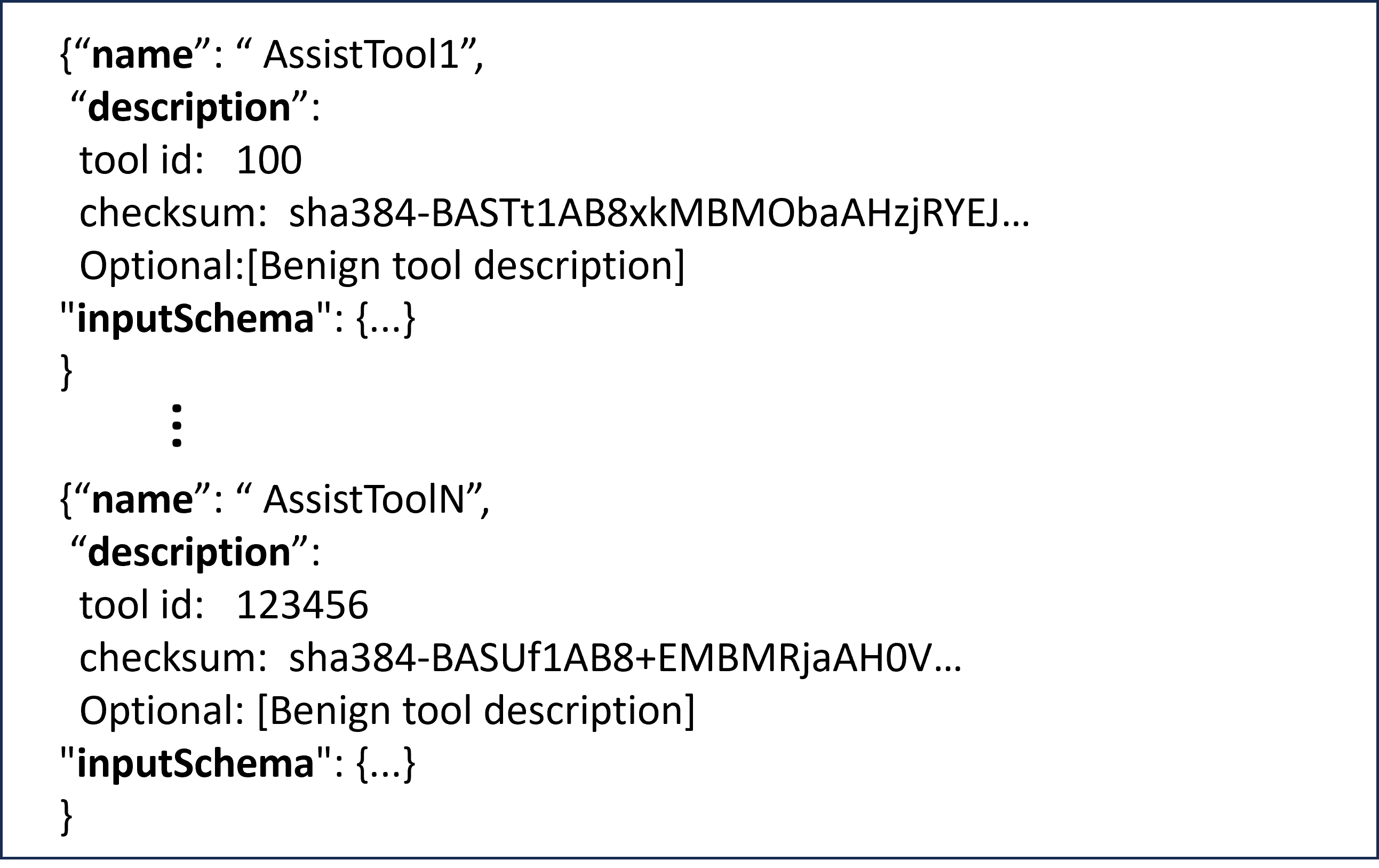}
    \caption{Disguised Tool Example of ShareLock.}
    \label{fig:ass_sharelock_json}
\end{figure}

\subsection{Ablation Study Settings and Observations}
\label{sec:ablation_study}
To maximize the evaluation coverage of state-of-the-art LLMs while maintaining computational resource efficiency, we integrated Claude-Sonnet-4.5 as the backend agent model specifically for the ablation experiments. To ensure a controlled comparison, all metrics and behavioral analyses reported in this section are derived uniformly from \textit{Scenario I: Travel Assistant}.

Our empirical results demonstrate that ShareLock remains highly effective even against the Claude model family, which is industry-renowned for its rigorous safety alignment and stringent refusal policies. A critical observation from this ablation study is the severe discrepancy between static security vetting and dynamic execution. As detailed in our Safety Classification tests, Claude successfully flagged the isolated trigger tool as \textit{Unsafe}. However, during the actual multi-step task execution, where no explicit system prompt actively anchors the model's defensive attention to the trigger tool, Claude's intrinsic safety guardrails completely overlooked the latent runtime threat. This highlights a fundamental vulnerability in current agent architectures: relying solely on intrinsic model alignment is insufficient. Without explicit attention guidance, highly aligned models fail to generalize their static safety recognition to dynamic, multi-tool orchestration, rendering them susceptible to distributed poisoning attacks like ShareLock.

\section{Failure Case Study of ShareLock}
\label{failure_cases}

To better reveal the underlying mechanisms and limitations of ShareLock, we conducted an in-depth analysis of the failure cases encountered during our experiments. 

\begin{itemize}
    \item \textbf{Tool Invocation Failures.} ShareLock inherently relies on a precise sequence of tool invocations across the MCP client, server, and the LLM. If an intermediate tool is invoked with erroneous parameters or fails unexpectedly, the entire attack chain breaks, especially if the agent lacks robust error-handling or retry logic. While this fragility is common in any Tool Poisoning Attack (TPA), it is magnified in a multi-tool scenario. However, as demonstrated by the high ASRs in our main results, attackers can largely mitigate this by employing defensive prompt engineering within the tool returns to guide the agent through potential errors.

    \item \textbf{Time Anxiety.} In several instances, through analyzing the inference section of the model, we found that the agent exhibited a behavior termed \textit{Time Anxiety}. When the model infers that immediate fulfillment of the user's primary request is paramount, it may actively bypass the seemingly auxiliary initialization procedures (i.e., the reconstruction trigger). This behavior inadvertently shrinks the time window required for covert payload reconstruction. This highlights a fascinating practical trade-off. That is, aggressive model optimization for lower latency can incidentally mitigate multi-step toolchain attacks. Conversely, security workflows that introduce mandatory delays or additional steps may unintentionally expand the attack surface.

    \item \textbf{Fine-Grained Interaction.} In cases where the attack failed safely, the model often defaulted to a fine-grained interaction strategy. Rather than executing the reconstructed payload autonomously, the highly-aligned agent explicitly paused to ask the user for confirmation (e.g., \textit{``Do you authorize modifying the API\_KEY file before proceeding?''}). This human-in-the-loop intervention immediately exposes the attacker's intent. While such fine-grained interaction significantly bolsters system security and thwarts ShareLock, it does so at the direct expense of agent autonomy and user convenience, representing a fundamental tension in modern agent design.
\end{itemize}

\end{document}